\documentclass[a4paper, USenglish, cleveref, autoref, thm-restate]{lipics-v2021}

\usepackage{comment}
\excludecomment{shortversion}
\includecomment{longversion}

\usepackage{hyperref}
\usepackage{amsthm}

\usepackage{tikz}
\usepackage{float}
\usetikzlibrary{positioning}
\usepackage{tabularx}
\usepackage{booktabs}
\usepackage{makecell}
\usepackage{bm}
\usepackage{xcolor}
\usepackage{mathtools}
\usepackage{marginnote}
\usepackage{ifthen}
\usepackage{derivative}
\usepackage{soul}
\usepackage{paralist}

\definecolor{cb-blue}{HTML}{648FFF}
\definecolor{cb-magenta}{HTML}{785EF0}
\definecolor{cb-pink}{HTML}{DC267F}
\definecolor{cb-red-orange}{HTML}{FE6100}
\definecolor{cb-orange}{HTML}{FFB000}

\definecolor{ALBG}{rgb}{1.0,0.8,0.8}

\usetikzlibrary{graphs,graphs.standard,calc}
\usetikzlibrary{positioning}
\usetikzlibrary{matrix}
\usetikzlibrary{decorations.pathreplacing,calligraphy,backgrounds}

\renewcommand{\theequation}{\roman{equation}}

\usepackage{pgfplots}
\pgfplotsset{compat=1.16}
\usepgfplotslibrary{fillbetween}
\usepgflibrary{shadings}
\usepgfplotslibrary{statistics}
\usepgfplotslibrary{colorbrewer}
\pgfplotsset{colormap/PuOr-11}
\usepackage{pgfplotstable}
\pgfplotsset{grid style={dotted,gray}}
\pgfplotsset{legend style={column sep=0.15cm}}

\newcommand{\select}[3]{
    \pgfplotstablegetelem{\coordindex}{#1}\of{#3}
    \ifnum\pgfplotsretval=#2
    \else
    \def\pgfmathresult{}
    \fi
}

\let\epsilon\varepsilon

\usepackage{todonotes}

\pdfoutput=1 
\hideLIPIcs 

\title{Insights into \texorpdfstring{$\bm{(k, \rho)}$}{(k, ρ)}-shortcutting algorithms} 

\author{Alexander Leonhardt}{Goethe University Frankfurt, Germany}{aleonhardt@ae.cs.uni-frankfurt.de}{https://orcid.org/0009-0006-8263-6900}{}
\author{Ulrich Meyer}       {Goethe University Frankfurt, Germany}{umeyer@ae.cs.uni-frankfurt.de}{https://orcid.org/0000-0002-1197-3153}{}
\author{Manuel Penschuck}   {Goethe University Frankfurt, Germany}                              {mpenschuck@ae.cs.uni-frankfurt.de}{https://orcid.org/0000-0003-2630-7548}{Funded by the Deutsche Forschungsgemeinschaft (DFG) -- {ME~2088/5-1} (FOR~2975 --- Algorithms,  Dynamics,  and Information Flow in Networks).}


\authorrunning{A. Leonhardt, U. Meyer, M. Penschuck} 

\Copyright{Alexander Leonhardt, Ulrich Meyer, Manuel Penschuck} 

\relatedversion{} 

\begin{shortversion}
\relatedversiondetails{The full version of this article is available on arXiv under the same name~\cite{fullversion}}{} 
\end{shortversion}

\supplement{
		The source code is publicly available on GitHub\\\href{https://github.com/alleonhardt/k-rho-shortcutting}{https://github.com/alleonhardt/k-rho-shortcutting}.
		In addition the raw data set and artifacts are provided on \href{https://ae.cs.uni-frankfurt.de/public\_files/k\_rho\_shortcutting\_algorithms\_data.zip}{https://ae.cs.uni-frankfurt.de/public\_files/k\_rho\_shortcutting\_algorithms\_data.zip}.}

\begin{longversion}
    \nolinenumbers 
\end{longversion}

\definecolor{goetheblau}{cmyk}{1.00, 0.20, 0.00, 0.40}
\definecolor{newblack}{rgb}{0,0,0}
\colorlet{sidecolor}{newblack!60}
\definecolor{darkgreen}{rgb}{0,0.5,0}

\usepackage{xspace}

\def\NP{\ensuremath{\mathcal{NP}}}

\def\ie{i.e.\xspace}
\def\eg{e.g., \space}
\def\etal{et\,al.\xspace}
\def\problem#1{\textsc{#1}\xspace}
\def\probVC{\problem{Vertex Cover}}
\def\probVCNN{\textsc{Vertex Cover}\xspace}
\def\probMSP{\problem{$k\rho$-MSP}}
\def\probSSSP{\problem{SSSP}}
\def\probPSSSP{\problem{PSSSP}}
\def\set#1{\ensuremath{\left\{#1\right\}}}
\def\setc#1#2{\ensuremath{\left\{#1\,\mid\,#2\right\}}}
\def\algoDP{$k\rho$-DP\xspace}
\def\algoGreedy{$k\rho$\textsc{-Greedy}\xspace}
\def\algoDPPC{$k\rho$-DP-PC\xspace}
\def\algoDPSA{$k\rho$-DP-SA\xspace}
\def\algoDPPCMH{$k\rho$-DP-PC+MH\xspace}
\def\algoDPSAMH{$k\rho$-DP-SA+MH\xspace}
\def\algoDPAll{$k\rho$-DP-PC-SA+MH\xspace}

\def\algoDPStar{$k\rho$-DP*\xspace}
\def\Oh{\ensuremath{\mathcal{O}}}

\def\Vedges{\ensuremath{V_\text{edges}}\xspace}

\def\Gsub{\ensuremath{G_P}\xspace}
\def\Vsub{\ensuremath{V_P}\xspace}
\def\Esub{\ensuremath{E_P}\xspace}

\def\Gtrans{\ensuremath{G_T}\xspace}
\def\Vtrans{\ensuremath{V_T}\xspace}
\def\Etrans{\ensuremath{E_T}\xspace}

\let\Go\Gtrans
\def\Gop{\ensuremath{G_\text{MSP}}\xspace}
\def\pitchfork{\textsc{pitchfork}\xspace}
\def\pitchforks{\textsc{pitchfork}s\xspace}

\def\trivial{\textsc{canonical}\xspace}
\def\complex{\textsc{complex}\xspace}

\def\shortver#1{#1}
\def\longver#1{}
\begin{longversion}
    \def\shortver#1{}
    \def\longver#1{#1}
\end{longversion}
\def\shortlong#1#2{\shortver{#1}\longver{#2}}

\begin{document}
\maketitle
\begin{abstract}
	A graph is called a $(k, \rho)$-graph iff every node can reach $\rho$ of its nearest neighbors in at most $k$ hops.
	This property proved useful in the analysis and design of parallel shortest-path algorithms~\cite{10.1145/2935764.2935765,delta-star-stepping}.
	Any graph can be transformed into a $(k, \rho)$-graph by adding shortcuts.
	Formally, the \probMSP problem asks to find an appropriate shortcut set of minimal cardinality.
	
	We show that $\probMSP$ is \NP-complete in the practical regime of $k \ge 3$ and $\rho = \Theta(n^\epsilon)$ for $\epsilon > 0$.
	With a related construction, we bound the approximation factor of known \probMSP heuristics~\cite{10.1145/2935764.2935765} from below and propose algorithmic countermeasures improving the approximation quality.
	Further, we describe an integer linear problem (ILP) solving \probMSP optimally.
	Finally, we compare the practical performance and quality of all algorithms in an empirical campaign.
\end{abstract}

\longver{\clearpage}

\section{Introduction}
Shortest path algorithms trace back to the very roots of computer science and are a part of the basic algorithmic toolbox.
Consequently a large body of literature has been devoted to shortest-path algorithms dating back to at least the 1950s with the classical single-source shortest-path (\probSSSP) algorithms due to Dijkstra~\cite{dijkstra1959note} as well as Bellman and Ford~\cite{Bel58,P-923}.

Despite this immense attention, shortest-path algorithms remain notoriously hard to parallelize efficiently;
one of the first parallel \probSSSP (\probPSSSP) is based on the Bellman-Ford algorithm.
More efficient solutions often follow the \emph{stepping framework} including $\Delta$-stepping~\cite{Meyer2003}, \textsc{radius-stepping}~\cite{10.1145/2935764.2935765} and $\Delta$*-stepping~\cite{delta-star-stepping}.
Although they outperform Bellman-Ford in practice, their theoretical guarantees on general graphs seldom reflect this.

Recently, Blelloch~\etal \cite{10.1145/2935764.2935765} introduced the notion of $(k,\rho)$-graphs (see \cref{definition:k-rho-ball}) which, roughly speaking, asserts that every node can reach its $\rho$ nearest neighbors by a shortest-path with at most $k$ edges.
This notion provides just enough structure to derive new theoretical bounds~\cite{10.1145/2935764.2935765,delta-star-stepping}.
Crucially, any graph can be converted into a $(k,\rho)$-graph by introducing shortcuts.
We refer to the problem of finding such a shortcut set of minimal cardinality as \probMSP (see \cref{def:probmsp}).

To the best of our knowledge, only two articles considered $(k, \rho)$-graphs previously.
Blelloch~\etal~\cite{10.1145/2935764.2935765} introduce the notion, propose two \probMSP heuristics, and empirically study the number of edges added by them.
Further, Dong~\etal~\cite{delta-star-stepping} rely on $(k, \rho)$-graphs to provide faster algorithms and new bounds without further investigating the notion itself.

\paragraph*{Our contribution.} 
We prove for the first time that \probMSP is \NP-complete for \emph{all} constant $k\geq 3$ and some $\rho=\Theta(n^\epsilon)$ in \cref{section:kp-msp-np-hard}.
This is of great practical relevance as \probPSSSP algorithms for $(k,\rho)$-graphs crucially depend on small $k$ and large $\rho$ for their span bounds.
As a by-product, we obtain non-trivial lower bounds on the approximation ratio of the best known \probMSP heuristic in \cref{sec:lower_bounds}.

Based on the insights derived from the lower bounds, we propose several new preprocessing steps in \cref{sec:preprocessing} to improve the heuristic's approximation factor on several random graph models.
In \cref{section:optimal-algorithm}, we give an integer linear program (ILP) to solve \probMSP optimally.
Finally, we conduct an experimental evaluation providing evidence that a commonly found property in social network graphs is at fault for frequently inducing large approximation factors for the existing heuristics.
Our newly proposed heuristics are able to mitigate this partially and a final assessment on real world graphs show that our techniques generalize well to a variety of graphs.

\section{Preliminaries}\label{sec:preliminaries}
Let $G=(V,E)$ be an undirected graph with vertex set~$V$ and edge set~$E$.
A \emph{weighted} graph has an additional weight function $W\colon E\to\mathbb{R}_{\geq 0}$, assigning each edge~$e$ the weight~$W(e)$.
Unweighted graphs imply the unit edge weights of~$1$.
We denote the \emph{neighbors} of $u \in V$ as $N(u) = \setc{v}{\set{u,v} \in E}$.
Further, let $G[V']$ be the subgraph of $G$ induced by the vertices~$V'$.

For two nodes $u, v \in V$, let $d(u,v)$ be the \emph{total weight of the shortest path} (by weight) connecting $u$ and $v$.
Further, let $\hat d(u,v)$ be the \emph{hop distance} between $u$ and $v$, which we define as the smallest number of edges among all paths between $u$ and $v$ with a weight of $d(u,v)$.
If no path exists between $u$ and $v$, set $d(u,v) = \hat d(u,v) = \infty$.
For the sake of brevity we denote by \emph{shortest path} a path with weight $d(u,v)$ and by \emph{closest path} we denote a path with weight $d(u,v)$ and $\hat d(u,v)$ edges.

\medskip
\noindent
We first restate the two central properties \emph{$k$-radius} and \emph{$\rho$-distance} according to~\cite{10.1145/2935764.2935765}:

\begin{definition}
	For $u \in V$, the \emph{$k$-radius $\bar{r}_k(u)$ of $u$} is the distance to the closest (by weight) node from $u$ which is more than $k$ hops away from $u$, \ie $\bar{r}_k(u) = \min_{v \in V\colon\ \hat d(u,v)>k}d(u,v)$.
\end{definition}

\begin{definition}
	For $u \in V$, the \emph{$\rho$-nearest distance of $u$}, denoted by $r_\rho(u)$, is the distance from $u$ to the $\rho$-th closest vertex to $u$.
\end{definition}

These allow us to define $(k, \rho)$-graphs, where each node can reach any of its $\rho$ nearest neighbors in at most $k$ hops.
Observe that we deviate from~\cite{10.1145/2935764.2935765} (which originally requires $r_\rho(v) \leq \bar{r}_k(v)$ for $(k,\rho)$-balls) to avoid ambiguity if multiple nodes lie at hop distance $k + 1$.

\begin{definition}\label{definition:k-rho-ball}
	Vertex $v \in V$ has a \emph{$(k,\rho)$-ball} iff $r_\rho(v) < \bar{r}_k(v)$; \ie the $\rho$-closest node of $v$ can be reached within $k$ hops.
	Further, $G$ is a $(k,\rho)$-graph iff every node has a $(k,\rho)$-ball.
\end{definition}

Arbitrary graphs can be transformed into $(k, \rho)$ graphs by adding shortcuts, \ie edges that lower the hop distance~$\hat d$ between nodes without changing their shortest path distance $d$:

\begin{definition}
	For $G=(V,E)$, let $u, v \in V$ be different nodes with distance $D = d(u,v) < \infty$ and $\hat d(u,v)>1$.
	\emph{A shortcut between $u$ and~$v$} is a new edge $\set{u,v}$ with weight $D$.
\end{definition}

\begin{definition}
	Given a weighted graph $G=(V,E)$, the $(k,\rho)${\normalfont\textsc{-Minimum-Shortcut Problem}} (\probMSP) asks for a minimum cardinality set $S$, s.t. $G'=(V,E\cup S)$ is a $(k,\rho)$-graph.
	Further let $k\rho\ell$-{\normalfont\textsc{MSP}}, the decision variant of \probMSP, ask whether $|S| \le \ell$ exists.
	\label{def:probmsp}
\end{definition}

\noindent
If clear from context, we use \probMSP as a shorthand for it decision variant $k\rho\ell$-\textsc{MSP}.

\section{Complexity of KP-MSP}\label{section:kp-msp-np-hard}
The \textsc{Minimum-Shortcut Problem} (\probMSP) can be solved in polynomial time for $k = 1$, since each shortcut only affects its two endpoints.
Hence, each node has to be connected to $\rho$ of its nearest neighbors.
In the following, we show that \probMSP is \NP-complete for $k > 2$ and practical $\gamma$.
Since it is trivial to verify a \probMSP solution in polynomial time, \probMSP is in \NP.
We therefore focus on its hardness by establishing that the \NP-hard~\cite{DBLP:conf/coco/Karp72} \probVC problem is polynomial time reducible to \probMSP.

Recall that, given an undirected graph $G=(V,E)$, \probVC asks to find a minimal set of nodes $C \subseteq V$ for which at least one endpoint of every edge is in $C$.
Since $G$ is undirected and the transformation in \cref{subsec:vc_transform} uses only unit edge weights, our proof suggests, that the hardness of \probMSP does not stem from edge weights.

\begin{theorem} \label{theorem:kp-msp-hardness}
    Let $k \geq 3$ be integer constants, and $\rho = \Theta(n^\epsilon)$ for some constant $\epsilon > 0$.
    Then, $\probVC \leq_p \probMSP$ implying $\probMSP \in \NP\text{-hard}$.
\end{theorem}

\begin{proof}
    We transform a \probVC input~$G=(V, E)$ into a \probMSP graph~\Gtrans in \cref{subsec:vc_transform}.
    In \cref{subsec:solving_vc}, we show that the structure of~\Gtrans implies that there exist optimal so-called \emph{canonical} \probMSP solutions~(see \cref{cor:canonical_msp}).
    Then we establish a bijection between shortcuts in canonical \probMSP solutions to nodes in \probVC.

    For simplicity, the proof only considers $\gamma \ge 7|V| + 6$.
    We refer to \shortlong{the full version~\cite{fullversion}}{\cref{sec:appendix_blowup}} of this paper, for the acclaimed range of~$\gamma$.
    The extended proof is structurally similar to the present one, but uses a recursive application of the \textsc{pitchfork}-gadgets in \cref{subsec:vc_transform}.
\end{proof}

\subsection{Transforming Vertex Cover to KP-MSP}\label{subsec:vc_transform}
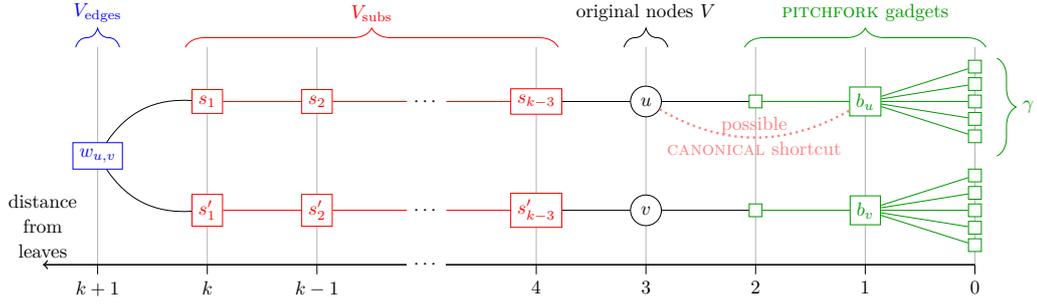
\begin{figure}[t]
    \centering
    \resizebox{\textwidth}{!}{
        \def\yspace{1}
        \def\ybar{-2}
        \begin{tikzpicture}[
                vertex/.style={draw, fill=white, circle, minimum width=1.6em, inner sep=0},
                subs/.style={draw, red, fill=white},
                edge vertex/.style={draw, blue, fill=white},
                gadget/.style={draw, green!60!black, fill=white},
                gedge/.style={draw, green!60!black},
                sedge/.style={draw, red}
            ]

            \foreach \x/\l in {0/{k+1},1/{k},2/{k-1}, 4/{4}, 5/{3}, 6/{2}, 7/{1}, 8/{0}} {
            \path[draw, black!30] (2*\x, -\ybar) to (2*\x, \ybar);
            \path[draw] (2*\x, \ybar) to node[pos=1,below] {$\l$} ++(0, -0.5em);
            }
            \path[draw, thick, ->] (16,\ybar) to node[pos=1, align=center, above] {distance \\ from \\ leaves} ++(-17, 0);
            \node[fill=white] at (6, \ybar) {$\cdots$};

            \node[edge vertex] (wuv) at (0,0) {$w_{u,v}$};
            \node[vertex] (u) at (10, \yspace) {$u$};
            \node[vertex] (v) at (10, -\yspace) {$v$};

            \foreach \s/\uv/\y in {{s}/u/\yspace, {s'}/v/-\yspace}{
                    \node[subs] (s-\uv-1) at (2, \y) {$\s_1$};
                    \node[subs] (s-\uv-2) at (4, \y) {$\s_2$};
                    \node[fill=white] (s-\uv-d) at (6, \y) {$\cdots$};
                    \node[subs] (s-\uv-k) at (8, \y) {$\s_{k-3}$};
                    \path[sedge] (s-\uv-1) to (s-\uv-2) to (s-\uv-d) to (s-\uv-k);

                    \node[gadget] (g-\uv-c) at (12, \y) {};
                    \node[gadget] (g-\uv-b) at (14, \y) {$b_\uv$};

                    \foreach \x in {0, 1,2,3,4} {
                            \node[gadget] (g-\uv-l\x) at (16, \y + 0.32*\x - 0.64) {};
                            \path[gedge] (g-\uv-b) to (g-\uv-l\x);
                        }

                    \path[gedge] (g-\uv-c) to (g-\uv-b);

                    \path[draw] (s-\uv-k) to (\uv) to (g-\uv-c);
                }

            \path[draw, bend left] (wuv) to (s-u-1);
            \path[draw, bend right] (wuv) to (s-v-1);

            \draw [blue, decorate,decoration={brace,amplitude=10pt},yshift=-0.1]
            (-0.4, -\ybar) to node[above, yshift=1em] {\Vedges} ++(0.8, 0);

            \draw [red, decorate,decoration={brace,amplitude=10pt},yshift=-0.1]
            (1.6, -\ybar) to node[above, yshift=1em] {$V_\text{subs}$} ++(6.8, 0);

            \draw [decorate,decoration={brace,amplitude=10pt},yshift=-0.1]
            (-0.4+10, -\ybar) to node[above, yshift=1em] {original nodes $V$} ++(0.8, 0);

            \draw [green!60!black,decorate,decoration={brace,amplitude=10pt},yshift=-0.1]
            (-0.2+12, -\ybar) to node[above, yshift=1em] {\pitchfork gadgets} ++(4.4, 0);

            \draw [green!60!black,decorate,decoration={brace,amplitude=10pt},yshift=-0.1]
            (16.4, \yspace + 0.8) to node[right, xshift=1em] {$\gamma$} ++(0, -1.8);

            \draw[very thick, dotted, red!50, bend right, align=center] (u) to node[below=-1.5em] {possible \\ \trivial shortcut} (g-u-b);

        \end{tikzpicture}
    }
    \caption{
        Transformation of edge $\set{u,v}$ for the input~$G=(V,E)$.
        Each edge in $E$ implies their own nodes in $\Vedges$ and $V_\text{subs}$ (left); each original node in $V$ has its own \pitchfork-gadget (right).
        By construction $w_{u,v}$ is too far from the leaves of either \pitchfork.
        By adding, \eg the \trivial shortcut (see \cref{def:trivial}) between $u$ and $b_u$, the leaves become available for the $(k, \rho)$-ball of $w_{u,v}$.
    }
    \label{fig:trans1}
\end{figure}

Let $G=(V, E)$ be an input graph for \probVC.
For a fixed value for parameter~$k \in [3, n)$, we transform $G$ into $\Go$ in two steps:

\begin{itemize}
    \item
          We first obtain $\Gsub = (\Vsub, \Esub)$ by replacing every edge $\set{u, v} \in E$ with its own copy of the following path template (\ie the subgraph induced by the following nodes is a path):
          \begin{align*}
              (u, \quad \underbrace{s_{k-3}, \ldots, s_1}_\text{$k-3$ subdivision}, \quad \underbrace{w_{u,v},}_\text{named representative of edge} \quad \underbrace{s'_{1}, \ldots, s'_{k-3}}_\text{$k-3$ subdivision}, \quad v)
          \end{align*}

          For $k=3$, the template degenerates into $(u, w_{u,v}, v)$.
          As visualized in \cref{fig:trans1}, we conceptually partition the vertex set $\Vsub$ into
          (i) the original nodes~$V$,
          (ii) the edge nodes $\Vedges = \setc{w_{u,v}}{\set{u,v} \in E}$, and
          (iii) the subdivision nodes $V_\text{subs}$ with $|V_\text{subs}| = 2|E|(k-3)$.
          Observe that this transformation retains the degrees of the original nodes~$V$ and that all new nodes $\Vsub \setminus V$ have degree 2.

    \item
          We obtain the final graph $\Gtrans = (\Vtrans, \Etrans)$ from \Gsub by adding and connecting a copy of a so-called $\gamma$-\pitchfork to each original node $v \in V$.
          A $\gamma$-\pitchfork to the host node $v \in V$ is a $(\gamma {+} 1)$-star graph where exactly one satellite has an additional edge to $v$.
          We denote the star's center of the gadget attached to node $v \in V$ as the \emph{base node}~$b_v$.
\end{itemize}

The $\gamma$-\pitchforks encode \probVC into the \probMSP problem.
The main idea is that all nodes but the ones in \Vedges have a $(k,\rho)$-ball.
The latter are too distant to the leaves of their respective $\gamma$-\pitchforks.
For any edge node $w_{u,v} \in \Vedges$, however, it suffices that the distance to the $\gamma$-\pitchfork of either $u$ or $v$ is reduced with a single shortcut.
This roughly corresponds to having at least one endpoint of each edge in a vertex cover.

\subsection{Solving \probVCNN}\label{subsec:solving_vc}
Let $G = (V, E)$ be the input graph for \probVC, $\Gtrans = (\Vtrans, \Etrans)$ the graph obtained from the transformation in \cref{subsec:vc_transform}, and $\Gop = (\Vtrans, \Etrans \cup S)$ where $S$ is a minimal cardinality shortcut set to ensure $\Gop$ is a $(k, \rho)$ graph.
We begin by formally establishing the observation that exactly the nodes $\Vedges$ have no $(k, \rho)$ ball in $\Gtrans$:
\begin{lemma} \label{lemma:sufficiently-large}
    Fix $\gamma\geq 7|V|+6$ and let $\rho=\gamma$.
    Then, $v \in \Vtrans$ forms a $(k,\rho)$-ball, iff $v \notin \Vedges$.
\end{lemma}
\begin{proof}
    Observe that by construction of~$\Gtrans$ and choice of $\gamma$, (i) the leaves of a single $\gamma$-pitchfork suffice to form a $(k, \rho)$-ball and (ii) any $(k, \rho)$-ball has to include the leaves of at least one $\gamma$-\pitchfork.
    As illustrated in \cref{fig:trans1}, all nodes except $\Vedges$ have a hop distance of at most $k$ to the leaves of their respective closest \pitchfork.
    Contrary, any edge node $w_{u,v} \in \Vedges$ requires $k+1$ hops to the leaves of the closest \pitchforks which are connected to nodes $u$ and $v$, respectively.
\end{proof}

For the remainder of this section, we assume that $\rho=\gamma$ and $\gamma\geq 7|V|+6$.
Then let~$S$ be a minimal cardinality \probMSP solution.
We now categorize the shortcuts into \trivial and \complex types, and show that any \complex shortcut can be canonicalized:

\begin{definition}\label{def:trivial}
    We call a shortcut \trivial iff it connects a node in $V$ to the base of its corresponding $\gamma$-\pitchfork (see \cref{fig:trans1}).
    All other shortcuts are called \complex.
\end{definition}

\begin{observation} \label{lemma:shortcut}
    A \trivial shortcut at node $u \in V$ allows all nodes $w_{u,v} \in \Vedges$ with $\set{u, v} \in E$ to form a $(k, \rho)$-ball.
\end{observation}

\begin{restatable}{lemma}{singleBallForComplex}
    \label{lemma:single_ball_for_complex}
    Let $S$ be an optimal solution for \probMSP on $\Gtrans$ and $s_c \in S$ be an arbitrary \complex shortcut.
    Then, the removal of $s_c$ destroys exactly one $(k, \rho)$ ball.
\end{restatable}

\begin{proof}[Proof sketch]
    \shortlong{The proof is included in the full version of this paper~\cite{fullversion}}{We refer to \cref{{sec:append_complex_shortcuts}} for the full proof}.
    Roughly speaking, it establishes two central properties of \complex shortcuts in a minimal shortcut set:
    \begin{itemize}
        \item No two \complex shortcuts interact in a meaningful way, \ie there is no shortcut $s_1$ that relies on a shortcut $s_2$ to close a $(k, \rho)$ ball.

        \item The distance between any two edge nodes in $\Vedges$ is sufficiently large, such that no \complex shortcut can affect more than one edge node~$v \in \Vedges$. \qedhere
    \end{itemize}
\end{proof}

\begin{corollary}\label{cor:canonical_msp}
    For any transformation~$G_T$, there exists an optimal \probMSP-solution~$S$ containing only \trivial shortcuts.
    We call such an~$S$ \emph{canonical}.
\end{corollary}

\begin{proof}
    Let $s \in S$ be an arbitrary \complex shortcut; if non exists, the claim follows.
    The removal of~$s$ destroys the $(k, \rho)$-ball of a unique~$w_{u,v} \in \Vedges$ due to \cref{lemma:single_ball_for_complex}.
    Thus we can replace~$s$ with the \trivial shortcut $s' = \set{u, b_u}$ (or ---equivalently--- $s' = \set{v, b_v}$).
    Observe that $S' = S \setminus \set{s} \cup \set{s'}$ has the same size $|S| = |S'|$, still turns \Gtrans into a $(k, \rho)$-graph, but has one fewer \complex shortcut.
    Finally recurse until $S'$ is canonical.
\end{proof}

\begin{lemma} \label{lemma:correctness}
    An optimal solution set for \probMSP on $\Go$ can be transformed into a solution for \textsc{vertex cover} on $G$.
\end{lemma}
\begin{proof}
    Let $C$ be an optimal \probVC solution on $G$.
    Then, by \cref{lemma:shortcut}, we know that $S = \setc{ \set{v, b_v}}{v \in C}$ consists only of \trivial shortcuts and turns $\Gtrans$ into a $(k, \rho)$-graph; \ie any optimal \probMSP solution~$S$ on \Gtrans satisfies $|S| \le |C|$.

    Let $S$ be an optimal solution of \probMSP on \Gtrans.
    Due to \cref{cor:canonical_msp}, assume without loss of generality that $S$ is canonical.
    Then, $C = \setc{u}{\set{u, b_u} \in S}$ is a \probVC (see \cref{lemma:shortcut}).
    Since $|S| \le |C|$, $C$ is minimal.
\end{proof}

\section{A rigorous bound on on the approximation ratio of \texorpdfstring{$\bm{k\rho}$}{kρ}-DP}
\label{sec:lower_bounds}
In \cref{section:kp-msp-np-hard}, we show that computing a minimal shortcut set~$S$ to convert an arbitrary graph into a $(k,\rho)$-graph can be expensive.
This has direct implications for preprocessing steps of algorithms based on $(k,\rho)$-graphs, which tend to favor large values of $\rho$.
For instance, the depth\footnote{
    The depth of a parallel algorithm is the length of its critical path, bounding the runtime from below --- even for an unbounded number of processors.
}
of \textsc{radius-stepping} on a $(k, \rho)$-graph with small constant $k$ is bounded by $\Oh(\frac{\log \rho}{\rho} \cdot
    n \log(n) L)$ where $L$ is the maximal edge weight.~\cite{10.1145/2935764.2935765}

Observe that any graph can be transformed into a $(k,\rho)$-graph by adding $\rho$ shortcuts per node, \ie $|S| \le n \rho$ is trivial.
Blelloch~\etal~\cite{10.1145/2935764.2935765} introduce two heuristics for smaller~$S$: the greedy \algoGreedy and \algoDP involving dynamic programming.
The authors show that \algoGreedy with $\rho=n-1$, $k=2$ suffers from an approximation ratio of at least $n-5$.
They demonstrate that, in practice, \algoDP yields smaller solutions than \algoGreedy.

Given a graph~$G=(V, E)$, the \algoDP heuristic~\cite{10.1145/2935764.2935765} uses a dynamic program to compute the minimal number of shortcuts needed to form a $(k, \rho)$-ball around a node~$s \in V$.
While the solution is optimal for this node conditioned on a shortest path tree (\cref{fig:not-one-short} depicts why the conditioning is necessary), \algoDP may not find a global minimal $|S|$, since the algorithm treats all $(k, \rho)$-balls independently (and in parallel).

For a fixed start node~$s \in V$ and let $T_s$ be a shortest path tree to the $\rho$-nearest neighbors of~$s$, such that for all $v \in T_s$ the path between $s$ and $v$ in $T_s$ has the least hops possible in~$G$.
If the height of $T_s$ is at most $k$, $s$ already has an $(k, \rho)$-ball; otherwise, we need to add shortcuts.
Consider some node~$u \in T_s$ with $u \ne s$, let $p$ be its parent, and denote $p$'s depth as $t = \hat d(s, p)$.
Then, define $F(u,t)$ as the smallest number of shortcuts into the subtree (of $T_s$) rooted in $u$ required to put $u$ and its children within the $(k, \rho)$-ball of $s$:
{\small
\begin{align*}
    F(u,t) = \min\Bigg(
    \overbrace{1+\sum_{w \in N^+(u)} F(w,1)}
    ^{\substack{\text{add shortcut $\set{s, u}$:} \\ \text{depth of $u$ is now $1$}}}, \ \
    \overbrace{\sum_{w \in N^+(u)} F(w,t+1)}
    ^{\substack{\text{no shortcut added to $u$}   \\ \text{depth of $u$ remains $t+1$}}}
    \Bigg)
    \quad \text{if $t < k$ else } \quad
    F(u, t)  = \infty
\end{align*}
}

\noindent
Summing over the children of $s$, $\sum_{u \in N(s)} F(u,0)$, yields the number of shortcuts on $s$.
To the best of our knowledge, we now bound the solution quality of \algoDP for the first time:

\begin{theorem} \label{theorem:approximation-factor}
    The approximation ratio of \algoDP on a graph with $\Theta(n)$ nodes is at least $n-1$ and for \algoGreedy $\Omega(n^2)$ for any constant $k\geq 3$ and some parameter $\rho$.
\end{theorem}
\begin{proof}
    Consider a star graph $S_n$ where each edge is replaced by a path of length $k-3$ as previously described in the reduction and add a single $\gamma$-\textsc{pitchfork} at the center node.
    In this graph only the satellite nodes do not form $(k,\rho)$-balls for some $k$ and some $\rho$.
    Clearly, by \cref{lemma:shortcut}, a single shortcut at the only $\gamma$-\textsc{pitchfork} is enough to transform this graph into a $(k,\rho)$-graph.
    \algoDP will add $n-1$ shortcuts as the only shortcuts considered start at the nodes themselves, thus preventing it to shortcut the $\gamma$-\textsc{pitchfork} gadget as any node within, by definition, already forms a $(k,\rho)$-ball.
    Hence, the aforementioned construction yields a graph of order $n+(n-1)\cdot (k-3)+2+\gamma=\Theta(n)$ which can be transformed into a $(k,\rho)$-graph by a single shortcut while \algoDP adds $n-1$.
    On the same construction the greedy heuristic has an approximation factor of $\Omega(n^2)$, every satellite node will add $\gamma-(k-2+2\cdot n)=\Theta(n)$ shortcuts thus resulting in $\Omega(n^2)$ added shortcuts.
\end{proof}

\section{The \texorpdfstring{$\bm{k\rho}$-DP-*}{kρ-DP-*} family of heuristics}
\label{sec:preprocessing}
Based on the insights leading to the lower bounds of \algoDP, we introduce two simple preprocessing steps for \algoDP, namely \algoDPPC (\underline{\textbf p}air-short\underline{\textbf c}utting) and \algoDPSA (\underline{\textbf s}et \underline{\textbf a}lignment).
These new heuristics address two separate issues introduced by a central property of \algoDP: All proposed shortcuts start at the source vertex on which \algoDP was invoked.

This property rules out any shortcuts starting from nodes which already form $(k,\rho)$-balls even if they are globally beneficial, such as \trivial shortcuts within a $\gamma$-\textsc{pitchfork} gadget.
We address this issue with \algoDPPC.
The new heuristic considers a larger set of shortcut candidates and introduces a global pooling phase to prune unpromising ones.

Additionally, since one end point of each shortcut introduced by \algoDP is determined by the current source vertex, there is little to no overlap between candidates of different source vertices.
Denote the set of shortcuts found for source vertex~$i$ by $S_i$. 
Then, for directed graphs for any $i \ne j$, the local solutions $S_i$ and $S_j$ are disjoint; for undirected graphs, each shortcut appears at most twice, \ie  $\sum_i |S_{i}| \leq 2 |S_{1}\cup S_{2}\cup \cdots \cup S_{n}|$.
Hence, we cannot expect significant savings from detected direct synergies between local solutions. 
Instead, our novel \algoDPSA generates a set of local (minimal) shortcut candidates by perturbing a single solution derived from \algoDP to maximize the inter-set alignments.
In \textbf{some} cases, our proposed algorithm allows the transformation of a shortcut set with $\Theta(n)$ distinct shortcuts to a global shortcut set of constant size.

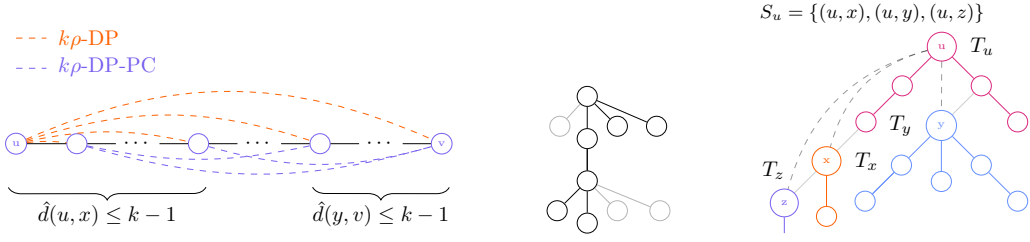
\begin{figure}
    \begin{subfigure}[t]{0.47\textwidth}
    \centering
    \scalebox{0.8}{
\begin{tikzpicture}[vertex/.style={draw,circle}, bl-edge/.style={color=cb-red-orange,dashed}, bl-pc-edge/.style={color=cb-magenta, dashed}]
    \node[draw,circle, color=cb-magenta, inner sep=0.07cm] (a0) at (0,0) {\tiny u};
    \foreach \i in {1,...,6} {
        \ifthenelse{\equal{\i}{2} \OR \equal{\i}{4} \OR \equal{\i}{6}}
        {\node (a\i) at (\i,0) {$\cdots$};}
        {

            \node[vertex,color=cb-magenta] (a\i) at (\i,0) {};
        }
    }
    \node[vertex, color=cb-magenta, inner sep=0.07cm] (a7) at (7,0) {\tiny v};

    \foreach \i in {0,...,6} {
        \pgfmathtruncatemacro\nexti{\i+1}
        \draw (a\i) -- (a\nexti);
    }

    \foreach \i in {1,...,7} {
        \ifthenelse{\equal{\i}{2} \OR \equal{\i}{4} \OR \equal{\i}{6}}
        {}
        {\draw[bl-edge,bend left=(\i+1)*3] (a0) to[->,> = latex] node[midway,above=0.2cm] (a0-a\i) {} (a\i);}
    }

    \foreach \i in {5,7} {
        \draw[bl-pc-edge,bend right=(\i+1)*2] (a1) to[->,> = latex] node[midway,below=0.2cm] (a1-a\i) {} (a\i);
    }

    \foreach \i in {7} {
        \draw[bl-pc-edge,bend right=(\i+1)*2] (a3) to[->,> = latex] node[midway,below=0.2cm] (a3-a\i) {} (a\i);
    }

    \node[align=left] (kp-dp) at (1.5,1.5) {{\color{cb-red-orange}\algoDP}\\{\color{cb-magenta}\algoDPPC}};
    \draw[color=cb-magenta,dashed] (0,1.25) -- (0.5,1.25);
    \draw[color=cb-red-orange,dashed] (0,1.75) -- (0.5,1.75);

    \draw[decorate, decoration={brace,amplitude=6pt,mirror,raise=17pt}] (a0.south west) -- node[below=0.7cm] {$\hat d(u,x) \leq k-1$} (a3.south east);
    \draw[decorate, decoration={brace,amplitude=6pt,mirror,raise=17pt}] (a5.south west) -- node[below=0.7cm] (lft) {$\hat d(y,v) \leq k-1$} (a7.south east);
\end{tikzpicture}}
\caption{The additional shortcuts considered by \algoDPPC in comparison to \algoDP for any pair of nodes $(u,v)$ for which $v$ is not within the $(k,\rho)$-ball of $u$.}
    \end{subfigure}
    \hfill
    \begin{subfigure}[t]{0.22\textwidth}
    \centering
    \scalebox{0.8}{
    \begin{tikzpicture}[vertex/.style={draw,circle},node distance=0.7cm, faded/.style={opacity=0.3}]
        \node[vertex] (b) {};

        \node[vertex,below right of=b] (b1) {};
        \node[vertex,right=0.3cm of b1] (b2) {};
        \node[vertex,faded,left=0.6cm of b1] (b3) {};
        \node[vertex,below of=b] (c) {};
        \node[vertex,below of=c] (d) {};
        \node[vertex,below of=d] (d1) {};
        \node[vertex,below left of=d] (d0) {};
        \node[vertex,faded,below right of=d] (d2) {};
        \node[vertex,faded,right of=d2] (d3) {};

        \draw (b) -- (c);
        \draw (c) -- (d);
        \draw (c) -- (d);
        \draw (d0) -- (d);
        \draw (d1) -- (d);
        \draw[faded] (d2) -- (d);
        \draw[faded] (d3) -- (d);
        \draw (b1) -- (b);
        \draw (b2) -- (b);
        \draw[faded] (b3) -- (b);
    \end{tikzpicture}}
    \caption{At most two \emph{leafs} of any node in a closest path tree are considered for the score.}
    \label{fig:important-breadth}
    \end{subfigure}
    \hfill
    \begin{subfigure}[t]{0.29\textwidth}
    \centering
    \scalebox{0.8}{
    \begin{tikzpicture}[vertex/.style={draw,circle},node distance=0.5cm]
    \node[color=cb-pink,vertex] (u) {\tiny u};
    \node[right=0.1cm of u] (ulabel) {$T_u$};
    \node[above left=0.1cm and -1cm of u] (ulabel2) {\small$S_u=\set{(u,x),(u,y),(u,z)}$};
    \node[vertex,color=cb-pink,below left=of u] (u0) {};
    \node[vertex,color=cb-pink,below right=of u] (u1) {};

    \node[vertex, color=cb-blue, below left=of u1] (u10) {\tiny y};
    \node[left=0.1cm of u10] (u10label) {$T_y$};
    \node[vertex,color=cb-pink,below right=of u1] (u11) {};

    \node[vertex,color=cb-blue, below left=of u10] (u100) {};
    \node[vertex,color=cb-blue, below right=of u10] (u101) {};

    \node[vertex,color=cb-blue, below=of u10] (u1000) {};
    \draw[color=cb-blue] (u10) -- (u1000);
    
    \node[vertex,color=cb-blue, below left=of u100] (u1001) {};
    \draw[color=cb-blue] (u100) -- (u1001);

    \node[vertex,color=cb-blue,below right=of u101] (u1011) {};

    \node[vertex,color=cb-pink,below left=of u0] (u00) {};

    \node[vertex,color=cb-red-orange,below left=of u00] (u000) {\tiny x};
    \node[right=0.1cm of u000] (u000label) {$T_x$};
    \node[vertex,color=cb-red-orange,below=of u000] (u00m) {};

    \draw[color=cb-pink] (u) -- (u0);
    \draw[color=cb-pink] (u) -- (u1);
    \draw[opacity=0.2] (u1) -- (u10);
    \draw[color=cb-pink] (u1) -- (u11);

    \draw[color=cb-blue] (u10) -- (u100);
    \draw[color=cb-blue] (u100) -- (u1001);
    \draw[color=cb-blue] (u10) -- (u101);

    \draw[color=cb-blue] (u101) -- (u1011);

    \draw[color=cb-pink] (u0) -- (u00);
    \draw[opacity=0.2] (u00) -- (u000);

    \draw[color=cb-red-orange] (u000) -- (u00m);

    \draw[bend right,dashed,opacity=0.4] (u) to (u000);
    \draw[dashed,opacity=0.4] (u) to (u10);

    \node[vertex,color=cb-magenta,below left=of u000] (u0000) {\tiny z};
    \draw[color=cb-magenta] (u0000) -- ++(0,-0.5cm);

    \node[above left=0.1cm and -0.3cm of u0000] (u000label) {$T_z$};
    \draw[opacity=0.2] (u000) -- (u0000);
    \draw[bend right,dashed,opacity=0.4] (u) to (u0000);
    \end{tikzpicture}}
    \caption{A part of the shortest path tree decomposition $T_u'$ for $u$ induced by the shortcut set $S_u$ derived from \algoDP.}
    \label{fig:algo-dp-dr}
    \end{subfigure}
    \caption{Schematics for the central components of \algoDPPC (a), (b) and \algoDPSA (c).}
\end{figure}

\subsection{Pair shortCutting: \texorpdfstring{$\bm{k\rho}$}{kρ}-DP-PC} \label{subsec:pc}
For every node pair $(u,v)$ with $u\in V$, $v \in N_{\rho}(u)$, and $ \hat d(u,v) > k$, consider a closest path from $u$ to $v$.
Let $S_{u,v}$ be the set of node pairs on this closest path such that any tuple $(x,y)$ in $S_{u,v}$ successfully asserts that a shortcut from $x$ to $y$ would move $v$ into the $(k,\rho)$-ball of $u$:
\begin{align} \label{definition:shortcut-set}
&&S_{u,v}=\set{(x,y)\ |\ \hat d(u,x)+1+\hat d(y,v) \leq k}
\end{align}

Observe that a candidate $(x,y) \in S_{u,v}$ can appear in other contexts~$S_{u', v'}$ as well.
Hence, we rate its global relative importance as the sum of local scores.
The local score of node $u$ is the reciprocal of the number of vertices that $u$ is still missing from its $(k,\rho)$-ball if only $(x,y)$ were inserted.
Then, we accept all candidates with scores exceeding $\mu+3\cdot\sigma$ where $\mu$ and $\sigma$ denote the mean and standard deviation of all scores within the shared hash map.
Finally, we run the original \algoDP heuristic on the resulting graph to form the remaining $(k,\rho)$-balls.

Notice that reducing $\rho$ with respect to $n$ reduces the number of interactions between distinct shortcuts.
To preserve the possibility of finding exceptional shortcuts we only include shortcuts where at least $k$ distinct vertices participated in its score and every of these distinct vertices has at least two nodes which are moved into their $(k,\rho)$-ball by the addition of the shortcut.
Furthermore we use the concept of \emph{important breadth} (\cref{fig:important-breadth}) since two leafs suffice so that our heuristic will give a higher score to the parent of both leafs than to any leaf individually, further considered leafs only exaggerate the importance of this shortcut while overshadowing shortcuts which decrease the depth \eg a $\gamma$-\textsc{pitchfork} gadget versus a path of length $2k+1$.
In addition to that the usage of the average and standard deviation allows us to select only exceptional shortcuts to increase the probability that such a shortcut improves the final solution.

The number of shortcuts processed for each node is bounded by $\Oh(\rho\cdot\sum_{i\leq k} i)=\Oh(\rho\cdot k^2)$.
All nodes can be processed in parallel and the global score can be shared by using a concurrent hash map supporting insert and update in expected $\Oh(1)$ time when $\rho=\Omega(\sqrt{n})$.
For $\rho=o(\sqrt{n})$ we have to fall back to a sorting algorithm followed by a prefix sum to avoid congestion problems for the shared hash map resulting in a worse span bound of $\Oh(\lceil\frac{\rho\cdot k^2}{\log\log n}\rceil(\log n)\log\log n)$.
Surrogate sums for both $\sigma$ and $\mu$ can be computed alongside the accumulation phase in $\Oh(1)$\footnote{We only need to keep track of two sums $\sum X^2$ and $\sum X$ to calculate both in the end in $\Oh(\log n)$. To keep track of the final $X$ for $X^2$ without knowing it we use  $X^2=1^2-1^2+...+(X-1)^2-(X-1)^2+X^2$.} per processed shortcut.
A final parallel iteration over the concurrent hash map containing at most $\max(n\cdot\rho\cdot k^2,(n-1)\cdot n)$ elements neither increases span nor work.
Thus we derive the following span bounds of \algoDPPC; $\Oh(\rho\cdot k^2)$ for $\rho=\Omega(\sqrt{n})$ and $\Oh(\lceil\frac{\rho\cdot k^2}{\log\log n}\rceil(\log n)\log\log n)$ otherwise.
For $k=\Oh(\sqrt{\log\rho})$ and $\rho=\Omega(\sqrt{n})$ this bound matches the one of \algoDP.
Notice that even in the regime $\rho=o(\sqrt{n})$, the worse span bounds are still below the span bounds for all known algorithms for \probPSSSP on $(k,\rho)$-graphs~\cite{delta-star-stepping}.

\subsection{Set Alignment: \texorpdfstring{$\bm{k\rho}$}{kρ}-DP-SA} \label{subsec:sa}
Given a graph $G=(V,E)$, let $S = S_1\cup S_2 \cup \cdots \cup S_n$ be a shortcut set, where $S_i$ denotes shortcuts starting in node $i$ as computed by \algoDP.
Then, for each node~$i$, we can decompose its shortest path tree~$T_i$ as illustrated in \cref{fig:algo-dp-dr} into separate subtrees induced by $S_i$.
For each $(i,x) \in S_i$ define the subtree $T_{i,x}$ rooted in $x$ such that all leafs of $T_{i,x}$ are either leafs of $T_i$ or the parents of a node $y\colon (i,y) \in S_i, y\neq x$.
Now, let $A=(v_1=i,v_2, \ldots ,v_k=x)$ denote a closest path in $G$ connecting $i$ to $x$ and let $A_S := \set{v_z\ |\ v_z \in A\land (i, v_z) \in S_i}$ be the set of nodes on this path which are the target of a shortcut from $i$.
For every subtree $T_{i,x}$, where $\textnormal{depth}(T_{i,x})< k-1$ we can construct a set $S_i':= S_i\cup\set{(t,x)} \setminus\set{(i,x)}$ where $t \in \set{v_{z+\delta}\ |\ v_z \in A_S\cup\set{v_1}, 0\leq \delta < k-1-\textnormal{depth}(T_{i,x})}$ which also constitutes a minimal shortcut set for node $i$.

Hence for any node $i$, we build the set $S_D:=\bigcup_{x,t} S_i'\cup S_i$ of distinct shortcuts being part of an optimal solution for $i$ and increase the score of such a shortcut by one in a concurrent shared hash map\footnote{Since we only need to know if a shortcut is shared by at least two nodes there is no write congestion, hence $\Oh(1)$ in expectation.}.
Any entry in it displaying a score of more than one signifies a shortcut that is in the intersection of at least two locally optimal solutions of unique nodes.

The number of perturbations a single node $i$ can construct of its shortcut set is bounded by $\Oh(|S_i|\cdot (k-1) \cdot \max_{A_S}|A_S|)$. By \cref{lemma:max-size} this is again bounded by $\Oh(\rho\cdot \max_{A_S}{|A_S|})$.
This increases our span bounds as $\max_{A_S}|A_S| = \Omega(\frac{\rho}{k})$ for some graphs.
However, by only considering a subset of $\log\rho$ predecessors for any shortcut, we can bound the span to $\Oh(\rho\log\rho)$ and work to $\Oh(n\rho\log\rho)$; this matches the performance of \algoDP.
Notice that in practice this limitation is often insignificant as many random graph classes (including Gilbert or hyperbolic random graphs) have with high probability bounded shortest paths lengths $\Oh(\log n)$ under reasonable assumptions on the edge weight distributions.
\begin{lemma} \label{lemma:max-size}
	Let $S_i$ be the shortcut set computed for node $i$ by \algoDP.
    Then, $|S_i| < \rho / (k-1)$.
\end{lemma}
\begin{proof}
    Let $T_i$ be the closest path tree rooted in $i$ containing the $\rho$-nearest neighbors. We decompose $T_i$ into $L_1,L_2, ... ,L_{k-1}$ where each $L_j := \setc{x}{x \in T_i,\ \hat d(i,x)\equiv j \mod (k-1)}$.
    Observe that $\sum_j |L_j| = |T_i|$, thus $\exists j: |L_j| \leq \frac{|T_i|}{k-1}$. The claim follows since the set $S_i' = \set{(i,x)\ |\ x \in L_j}$ is discoverable by \algoDP and forms a $(k,\rho)$-ball for $i$.   
\end{proof}

\subsection{MinHash} \label{subsec:minhash}
Both presented heuristics can still introduce redundant shortcuts originating from different local contexts; for instance, consider partially overlapping shortcuts along a path.
In the following, we propose a probabilistic approach based on  MinHash\cite{10.5555/829502.830043} to detect synergetic shortcuts that partially overlap.

\begin{longversion}
MinHash, defines a family of hash functions on a set such that for a fixed hash function $H$ the collision probability of two subsets matches their Jaccard index $J(A,B)=|A\cap B| / |A\cup B|$.
To reduce the variance of the estimator, typically, $c > 1$ such hash functions are used with $Z=\sum_{j=0}^{c-1} I_j/c$ where $I_j$ denotes the event of a collision of hashes of set $A$ and $B$ on hash function $H_j$.
Thus for \algoDPPC we represent a shortcut $(u,v)$ by 
\begin{align}
	M_{u,v} := \setc{(x,y)}{x, y \in V, \textnormal{a shortcut from $u$ to $v$ moves $y$ into the $(k,\rho)$-ball of $x$}}.
\end{align}
We keep a shared concurrent hash set for each hash function keeping track of the already encountered hashes\footnote{This changes the span bounds analogously to \algoDPPC due to possible write congestion.}.
Let $(x,y)$ be a shortcut under consideration and let $C$ be the set of all added shortcuts, then we can assure that for every $0 \le j < c$
\begin{align}
	\mathbb{P}(H_j(M_{x,y}) \textnormal{ collides}) \geq \max_{(a,b) \in B} J(M_{x,y},M_{a,b}).
\end{align}

This allows us to decrease redundancies by upper bounding the Jaccard index of nodes already profiting from previously introduced shortcuts with nodes that are meant to benefit from the new shortcut in $\Oh(1)$.\footnote{Since the hash is trivially derived from union operations $H(A\cup B) = \min(H(A),H(B))$, we can only keep the fingerprints to ensure the computation does not increase the span or work.}
For \algoDPSA, we represent a shortcut $(u,y)$ by $M_{u,y} := \set{x\ |\ x\in V,\hat d_{G'}(x,u)+1+\textnormal{depth}(T_{x,y})\leq k}$, the argument on the foundation is the same as for \algoDPPC.
\end{longversion}
\begin{shortversion}
    The omission of shortcuts displaying a high degree of overlap results in a marked redundancy reduction and in turn improves the solution size of our heuristics.
    For the missing details refer to \cref{sec:appendix_minhash}.
\end{shortversion}

\section{Optimal algorithm} \label{section:optimal-algorithm}
In this section, we present ---to the best of our knowledge--- the first exact solver for \probMSP.
Since we established the problem to be \NP-hard in \cref{section:kp-msp-np-hard}, there is little hope for a solution that is efficient on all instances.
Hence, we propose an integer linear programming (ILP) formulation to harness the extensive research into efficient ILP solvers.

Let $N_\rho(v) = \setc{u \in V}{u \in V,\ u \neq v,\ d(v,u) \leq r_p(v)}$ be the set of vertices with distance at most the $\rho$-nearest neighbor, and define the shortcut $N^+_\rho(v) = N_\rho(v)\cup \set{v}$. %
Recall that $N_\rho(v)$ can contain more than $\rho$ vertices in the case that there are several vertices with the same distance to $v$ as the $\rho$-closest vertex.
Define $w(u,v)$ as the weight of the edge $\set{u,v}$ if it exists or else the weight of a shortest path between $u$ and $v$.
Observe that in general $w(u,v)$ can exceed the shortest-path distance $d(u,v)$.
Finally, let $S$ be the set of possible shortcuts, \ie $S = \setc{(u,v)}{u \in V,\ v \in N_\rho(u), \ (u,v) \notin E}$.

We provide several variations of the ILP formulation for both the the undirected \emph{(U)} and the directed \emph{(D)} case of \probMSP.
In contrast to the directed case, edges are bidirectionally usable in the undirected case, thus requiring a distinct formulation to encode the altered constraints.
\begin{shortversion}
    An explicit description of the formulation, correctness and a discussion on the encoding size can be found in \cref{sec:appendix-optimal-algorithm}.
\end{shortversion}
\begin{longversion}
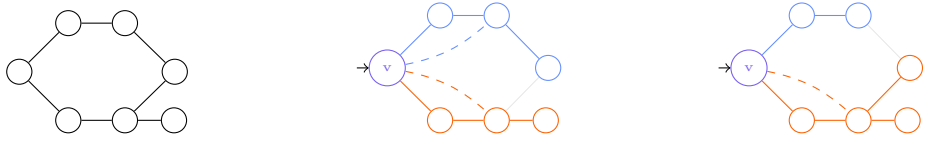
\begin{figure}
    \begin{subfigure}{0.32\textwidth}
        \centering
        \rotatebox{90}{
        \begin{tikzpicture}[dnode/.style={circle,draw}]
            \node[dnode] (a) {};
            \node[dnode,below left=0.4cm and 0.4cm of a] (b) {};
            \node[dnode,below right =0.4cm and 0.4cm of a] (c) {};

            \node[dnode,below=0.4cm of b] (bb) {};
            \node[dnode,below=0.4cm of c] (cb) {};

            \node[dnode,below=1.7cm of a] (ab) {};
            \node[dnode,below=0.3cm of bb] (bbb) {};

            \draw (a) -- (b);
            \draw (a) -- (c);

            \draw (b) -- (bb);
            \draw (c) -- (cb);

            \draw (bb) -- (bbb);
            \draw (bb) -- (ab);
            \draw (cb) -- (ab);
        \end{tikzpicture}}
    \end{subfigure}
    \hfill
    \begin{subfigure}{0.32\textwidth}
        \centering
        \rotatebox{90}{
        \begin{tikzpicture}[dnode/.style={circle,draw}, tree1/.style={color=cb-red-orange}, tree2/.style={color=cb-blue}]
            \node[color=cb-magenta,dnode] (a) {\rotatebox{-90}{\tiny v}};
            \draw[->] (a) ++ (0cm,0.4cm) to (a);
            \node[tree1,dnode,below left=0.4cm and 0.4cm of a] (b) {};
            \node[tree2,dnode,below right =0.4cm and 0.4cm of a] (c) {};

            \node[tree1,dnode,below=0.4cm of b] (bb) {};
            \node[tree2,dnode,below=0.4cm of c] (cb) {};

            \node[tree2,dnode,below=1.7cm of a] (ab) {};
            \node[tree1,dnode,below=0.3cm of bb] (bbb) {};

            \draw[tree1] (a) -- (b);
            \draw[tree2] (a) -- (c);

            \draw[tree1] (b) -- (bb);
            \draw[tree2] (c) -- (cb);

            \draw[tree1] (bb) -- (bbb);
            \draw[tree2] (cb) -- (ab);

            \draw[bend left=15pt, dashed, color=cb-red-orange] (a) to (bb);
            \draw[bend right=15pt, dashed, color=cb-blue] (a) to (cb);
            \draw[opacity=0.1] (bb) -- (ab);
        \end{tikzpicture}}
    \end{subfigure}
    \hfill
    \begin{subfigure}{0.32\textwidth}
        \centering
        \rotatebox{90}{
        \begin{tikzpicture}[dnode/.style={circle,draw}, tree1/.style={color=cb-red-orange}, tree2/.style={color=cb-blue}]
            \node[color=cb-magenta,dnode] (a) {\rotatebox{-90}{\tiny v}};
            \draw[->] (a) ++ (0cm,0.4cm) to (a);
            \node[tree1,dnode,below left=0.4cm and 0.4cm of a] (b) {};
            \node[tree2,dnode,below right =0.4cm and 0.4cm of a] (c) {};

            \node[tree1,dnode,below=0.4cm of b] (bb) {};
            \node[tree2,dnode,below=0.4cm of c] (cb) {};

            \node[tree1,dnode,below=1.7cm of a] (ab) {};
            \node[tree1,dnode,below=0.3cm of bb] (bbb) {};

            \draw[tree1] (a) -- (b);
            \draw[tree2] (a) -- (c);

            \draw[tree1] (b) -- (bb);
            \draw[tree2] (c) -- (cb);

            \draw[tree1] (bb) -- (bbb);
            \draw[tree1] (bb) -- (ab);

            \draw[bend left=15pt, dashed, color=cb-red-orange] (a) to (bb);
            \draw[opacity=0.1] (cb) -- (ab);
        \end{tikzpicture}}
    \end{subfigure}
    \caption{A graph for which the optimality of the solution of \probMSP for $k=2$, $\rho=6$ and a single node $v$ already depends on the encoding of a specific shortest path tree with fewest hops. Thus until it is clear which tree preserves the optimality of the solution an optimal algorithm has to encode all of them.}
    \label{fig:not-one-short}
\end{figure}
\end{longversion}

\begin{longversion}
\renewcommand{\theequation}{\arabic{equation}}
\begin{align}
    \textbf{min}      & \quad \sum_{(u,v) \in S} s_{u,v} \label{eq:ilp-objective},\ \textbf{subject to}\\
                         & \forall s \in V,\; \forall e \in N_\rho(s) \nonumber                                                                                                                                       \\
                         & \quad \sum_{v \in N_\rho(s)} x^{s,e}_{s,v} - \sum_{v \in N_\rho(s)} x^{s,e}_{v,s}                     &                                          & = 1 \label{eq:ilp-start}           \\
                         & \quad -\sum_{v \in N_\rho(s)} x^{s,e}_{v,e} - \sum_{v \in N_\rho(s)} x^{s,e}_{e,v}                    &                                          & = -1 \label{eq:ilp-end}            \\
                         & \quad \sum_{u,v \in N_\rho(s)} x^{s,e}_{u,v} - \sum_{u,v \in N_\rho(s)} x^{s,e}_{v,u}                 &                                          & = 0 \label{eq:ilp-flow}            \\
                         & \quad \sum_{u,v \in N^+_\rho(s)} x^{s,e}_{u,v}\cdot w(u,v)                                            &                                          & = d(s,e) \label{eq:ilp-distance}   \\
                         & \quad \sum_{u,v \in N^+_\rho(s)} x^{s,e}_{u,v}                                                        &                                          & \leq k+\hat d(s,e)\cdot u^{s,e} \label{eq:ilp-hop-distance} \\
                         & \quad \sum_{e \in N_\rho(s)} u^{s,e}                                                        &                                          & = |N_\rho(s)|-\rho \label{eq:ilp-at-most-rho} \\
    \nonumber                                                                                                                                                                                                         \\
                         & \forall u,v \in N_\rho(s)\colon u\neq v, s \in V, e \in N_\rho(s)
    \quad x^{s,e}_{u,v} &                                                                                                            & \geq 0                                                                        \\
    \text{(D)}\quad & \scalebox{1.3}{[}\forall s \in V,\; \forall e \in N_\rho(s), \forall x^{s,e}_{u,v}\colon (u,v) \notin E[G[N^+_\rho(s)]]
    \quad x^{s,e}_{u,v} &                                                                                                            & \leq s_{u,v} \label{eq:ilp-use-existing}\scalebox{1.3}{]}				      \\
    \text{(U)}\quad &{\scalebox{1.3}{[}} \forall s \in V,\; \forall e \in N_\rho(s), \forall x^{s,e}_{u,v}\colon (u,v) \notin E[G[N^+_\rho(s)]]
    \quad x^{s,e}_{u,v} &                                                                                                            & \leq s_{u,v}+s_{v,u}\scalebox{1.3}{]} \label{eq:ilp-use-existing-undirected}
\end{align}
\renewcommand{\theequation}{\roman{equation}}
\textbf{Correctness. }
\cref{eq:ilp-start,eq:ilp-end,eq:ilp-flow} are the standard\cite{TACCARI2016122} shortest path ILP formulation constraints, encoded for all pairs $(s,e): s \in V, e \in N_\rho(s)$.
Hence, this allows us to find a path between $s$ and $e$. Contrary to the standard formulation, we omit the optimization of the path weight but encode this as a separate constraint.
\cref{eq:ilp-distance} encodes the shortest path distance, such that the path connecting $s$ and $e$ is guaranteed to be a shortest path.
The constraint \cref{eq:ilp-hop-distance} further limits the number of traversed edges to at most~$k$ for any path to \emph{selected} neighbors.
This distinction is needed as the cardinality of $|N_\rho(s)|$ can be larger than $\rho$, thus only a subset of $N_\rho(s)$ has to be reached with at most $k$ edges.
Consequently \cref{eq:ilp-at-most-rho} induces with the binary variables $u^{s,e}$ a subset $N^{s} \subseteq N_\rho(s)$ with $|N^s|=\rho$ for which the constraints of \cref{eq:ilp-hop-distance} hold while for all other nodes $v\in N_\rho(s)\setminus N^s$ the constraints are trivially satisfied.
Finally the transitive encodings in \cref{eq:ilp-use-existing} and \cref{eq:ilp-use-existing-undirected} allow the unidirectional or bidirectional usage of any shortcut to fulfill the former constraints incurring a cost of one (see \cref{eq:ilp-objective}).
Observe that for the undirected case (\cref{eq:ilp-use-existing-undirected}) no further constraint $x^{s,e}_{u,v} \leq 1$ is needed as \cref{eq:ilp-distance} implicitly enforces this already.
An integrality constraint on all $s_{u,v}$ and $x^{s,e}_{u,v}$ allows direct recovery of a minimal cardinality shortcut set. 

\noindent
\textbf{Encoding size. }
We focus on the number of $x^{s,e}_{u,v}$ variables, since they dominate all other variable types due to the injective mapping in \cref{eq:ilp-use-existing} (or equivalently in \cref{eq:ilp-use-existing-undirected}).
The ILP formulation encodes the path finding for $\sum_{u\in V} |N_\rho(u)| = \Oh(n\cdot \max_{u\in V} |N_\rho(u)|)$ pairs in total.
For each pair, we need to encode a path of length at most $k$ from the set of all possible edges $E[G[N_\rho(s)]]$ in a locally full graph, thus implying a bound of $\Oh(|N_\rho(s)| \cdot |N_\rho(s)|)$ on the cardinality of $E[G[N_\rho(s)]]$.
Combining this yields a bound of $\Oh(n\cdot \max_{u\in V} |N_\rho(u)|^3) = \Oh(n^4)$ variables.

Improved bounds on the cardinality of $N_\rho(\cdot)$, directly translate into a sharper bound for the total number of variables.
For this consider the complete bipartite graph $K_{n/2,n/2}$ with unit edge weights, clearly for all nodes there are $n/2$ other nodes at distance one.
Thus in $K_{n/2,n/2}$ $\forall v \in V: |N_1(v)| = \Omega(n)$ and by this we cannot put better bounds on the cardinality of $N_\rho(\cdot)$ in the general case.
Contrary, we have $\forall v \in V\ |N_\rho(v)| = \Theta(\rho)$ if the sets of the $\rho$-nearest neighbors are unique for every node.
This happens for instance if all distances are distinct which is highly likely in graphs with random continuous edge weights.

At this point a further optimization by directly encoding a shortest path tree is not possible as seen in \cref{fig:not-one-short}.
The choice which shortest path tree to encode for some node $v$ highly influences the possible shortcuts that are available for the $(k,\rho)$-shortcutting algorithm and it is unclear which tree preserves the ability to find the optimal solution.
\end{longversion}

\pgfplotstableread[col sep=comma]{data/gilbert-ilp-analysed.csv}\tILP
\pgfplotstableread[col sep=comma]{data/gilbert-ilp-max-analysed.csv}\tILPMax

\pgfplotstableread[col sep=comma]{data/gilbert-ilp-feasable-analysed.csv}\tILPFeasable
\pgfplotstableread[col sep=comma]{data/gilbert-ilp-raw-all.csv}\tILPRawAll

\pgfplotstableread[col sep=comma]{data/powerlaw-ilp-analysed.csv}\tILPPowerlaw
\pgfplotstableread[col sep=comma]{data/powerlaw-ilp-max-analysed.csv}\tILPMaxPowerlaw

\pgfplotstableread[col sep=comma]{data/gilbert-ilp-raw-n-20.csv}\tILPRawtwenty
\pgfplotstableread[col sep=comma]{data/gilbert-ilp-raw-n-25.csv}\tILPRawtwentyfive
\pgfplotstableread[col sep=comma]{data/gilbert-ilp-raw-n-30.csv}\tILPRawthirty
\pgfplotstableread[col sep=comma]{data/gilbert-ilp-raw-n-35.csv}\tILPRawthirtyfive
\pgfplotstableread[col sep=comma]{data/gilbert-ilp-raw-n-40.csv}\tILPRawfourty
\pgfplotstableread[col sep=comma]{data/gilbert-ilp-raw-n-45.csv}\tILPRawfourtyfive

\pgfplotstableread[col sep=comma]{data/hyperbolic-ilp-analysed.csv}\tHyperbolicILP
\pgfplotstableread[col sep=comma]{data/hyperbolic.csv}\tHyperbolicILPRaw
\pgfplotstableread[col sep=comma]{data/hyperbolic-ilp-raw-all.csv}\tHyperbolicILPRawAll
\pgfplotstableread[col sep=comma]{data/hyperbolic-ilp-max-analysed.csv}\tHyperbolicMaxILP
\pgfplotstableread[col sep=comma]{data/hyperbolic-ilp-raw-n-20.csv}\tHyperbolicILPRawtwenty
\pgfplotstableread[col sep=comma]{data/hyperbolic-ilp-raw-n-25.csv}\tHyperbolicILPRawtwentyfive
\pgfplotstableread[col sep=comma]{data/hyperbolic-ilp-raw-n-30.csv}\tHyperbolicILPRawthirty
\pgfplotstableread[col sep=comma]{data/hyperbolic-ilp-raw-n-35.csv}\tHyperbolicILPRawthirtyfive
\pgfplotstableread[col sep=comma]{data/hyperbolic-ilp-raw-n-40.csv}\tHyperbolicILPRawfourty
\pgfplotstableread[col sep=comma]{data/hyperbolic-ilp-raw-n-45.csv}\tHyperbolicILPRawfourtyfive
\pgfplotstableread[col sep=comma]{data/hyperbolic-ilp-raw-n-50.csv}\tHyperbolicILPRawfifty
\pgfplotstableread[col sep=comma]{data/hyperbolic-ilp-raw-n-55.csv}\tHyperbolicILPRawfiftyfive
\pgfplotstableread[col sep=comma]{data/hyperbolic-ilp-raw-n-60.csv}\tHyperbolicILPRawsixty
\pgfplotstableread[col sep=comma]{data/hyperbolic-ilp-raw-n-65.csv}\tHyperbolicILPRawsixtyfive
\pgfplotstableread[col sep=comma]{data/hyperbolic-ilp-raw-n-70.csv}\tHyperbolicILPRawseventy

\pgfplotstableread[col sep=comma]{data/powerlaw-ilp-raw-n-20.csv}\tILPPowerlawRawtwenty
\pgfplotstableread[col sep=comma]{data/powerlaw-ilp-raw-n-25.csv}\tILPPowerlawRawtwentyfive
\pgfplotstableread[col sep=comma]{data/powerlaw-ilp-raw-n-30.csv}\tILPPowerlawRawthirty
\pgfplotstableread[col sep=comma]{data/powerlaw-ilp-raw-n-35.csv}\tILPPowerlawRawthirtyfive
\pgfplotstableread[col sep=comma]{data/powerlaw-ilp-raw-n-40.csv}\tILPPowerlawRawfourty
\pgfplotstableread[col sep=comma]{data/powerlaw-ilp-raw-n-45.csv}\tILPPowerlawRawfourtyfive
\pgfplotstableread[col sep=comma]{data/powerlaw-ilp-raw-n-50.csv}\tILPPowerlawRawfifty
\pgfplotstableread[col sep=comma]{data/powerlaw-ilp-raw-n-55.csv}\tILPPowerlawRawfiftyfive
\pgfplotstableread[col sep=comma]{data/powerlaw-ilp-raw-n-60.csv}\tILPPowerlawRawsixty
\pgfplotstableread[col sep=comma]{data/powerlaw-ilp-raw-n-65.csv}\tILPPowerlawRawsixtyfive
\pgfplotstableread[col sep=comma]{data/powerlaw-ilp-raw-n-70.csv}\tILPPowerlawRawseventy

\pgfplotstableread[col sep=comma]{data/results_comparison_gilbert_3.csv}\tLargeComparisonGilbertThree
\pgfplotstableread[col sep=comma]{data/results_comparison_gilbert_4.csv}\tLargeComparisonGilbertFour
\pgfplotstableread[col sep=comma]{data/results_comparison_gilbert_5.csv}\tLargeComparisonGilbertFive

\pgfplotstableread[col sep=comma]{data/results_comparison_powerlaw_3.csv}\tLargeComparisonPowerlawThree
\pgfplotstableread[col sep=comma]{data/results_comparison_powerlaw_4.csv}\tLargeComparisonPowerlawFour
\pgfplotstableread[col sep=comma]{data/results_comparison_powerlaw_5.csv}\tLargeComparisonPowerlawFive

\pgfplotstableread[col sep=comma]{data/results_comparison_hyperbolic_3.csv}\tLargeComparisonHyperbolicThree
\pgfplotstableread[col sep=comma]{data/results_comparison_hyperbolic_4.csv}\tLargeComparisonHyperbolicFour
\pgfplotstableread[col sep=comma]{data/results_comparison_hyperbolic_5.csv}\tLargeComparisonHyperbolicFive

\pgfplotstableread[col sep=comma]{data/real_world_proc_mean.csv}\tRealWorldDataProcessedMean
\pgfplotstableread[col sep=comma]{data/real_world_proc_2500.csv}\tRealWorldDataProcessedFirst
\pgfplotstableread[col sep=comma]{data/real_world_proc_2500_mean.csv}\tRealWorldDataProcessedFirstMean
\pgfplotstableread[col sep=comma]{data/real_world_proc_5000.csv}\tRealWorldDataProcessedSecond
\pgfplotstableread[col sep=comma]{data/real_world_proc_5000_mean.csv}\tRealWorldDataProcessedSecondMean
\pgfplotstableread[col sep=comma]{data/real_world_proc_7500.csv}\tRealWorldDataProcessedThird
\pgfplotstableread[col sep=comma]{data/real_world_proc_7500_mean.csv}\tRealWorldDataProcessedThirdMean
\pgfplotstableread[col sep=comma]{data/real_world_proc_10000.csv}\tRealWorldDataProcessedFourth
\pgfplotstableread[col sep=comma]{data/real_world_proc_10000_mean.csv}\tRealWorldDataProcessedFourthMean

\section{Experiments}
\renewcommand{\theequation}{\arabic{equation}}

In this section we compare the previously discussed algorithms on several random graph models.
The ILP encoding is implemented in Python invoking Gurobi\footnote{\url{https://gurobi.com}} 11.0.0.
We implemented all heuristics (including those proposed by Blelloch~\etal~\cite{10.1145/2935764.2935765}) in the Rust\footnote{Rust is a compiled language of comparable performance to C~\cite{9640225}.} .
If not stated otherwise the following standard parameters and considerations apply to all experiments:
\begin{itemize}
    \item We use several machines with AMD EPYC 7452 (64 threads) and 7702P (128 threads) CPUs, up to 512 GB RAM, \emph{Ubuntu  22.04.3}, the \emph{Python 3.10.12} and \emph{Rust 1.77.0}.
    \item We focus on the case $\rho=n{-}1$ as this confers two benefits: \begin{inparaenum} \item it is arguably the case with the most optimization potential, \item the difference in the observed solution is solely attributable to a superior choice of shortcuts since $N_\rho(v)$ is unique for every node $v$.\end{inparaenum}
    \item Each ILP instance has a timeout of 1800\,s after which it is deemed unsolvable.
\end{itemize}

We are using several random graph models.
The Gilbert model~\cite{Gilbert1959} functions as null model with almost no discernible structure.
We denote it as $\mathcal{G}(n,p)$ where $n$ is the number of nodes and $p$ the independent probability of each edge to be present.

One prominent structural feature found in many observed networks are powerlaw degree distributions (\ie a random node has degree $X$ with $\mathbb{P}[X = x] \propto x^{-\gamma}$).~\cite{barabasi2014network}
This is of special interest here, since our discussion in \cref{sec:lower_bounds} suggests that such skewed distributions are hard instances for the existing heuristics.
The hyperbolic random graph model~\cite{PhysRevE.82.036106} $\mathcal{H}(n,k,\gamma,T)$ with $\gamma=3.0$, $k=3$ and $T=0$ is known to yield sparse graphs featuring powerlaw degree distributions and a non-vanishing local cluster coefficient.~\cite{barabasi2014network}
Since the presence of a local community structure (\eg short cycles etc.) may affect the shortcutting behavior, we also consider a Markov Chain randomization model~\cite{DBLP:conf/alenex/GkantsidisMMZ03} $\mathcal{MC}(G,s)$.
It approximates a uniform sample from all simple graphs with prescribed degrees.
This model allows us to quantify the impact of the degree distribution mostly independently from other structural properties.

\subsection{Experimental evaluation}
\begin{figure}[t]
    \begin{subfigure}{0.34\textwidth}
        \scalebox{0.57}{
            \begin{tikzpicture}
                \pgfplotstablegetrowsof{\tHyperbolicILP}
                \pgfmathtruncatemacro\NodeRows{\pgfplotsretval-1}
                \begin{axis}[legend pos=north west, xmin=13,xmax=78, ylabel=average number of added shortcuts, xlabel=number of nodes,grid=both, title={$\mathcal{H}(n,k=3,\gamma=3.0,T=0)$},legend style={fill opacity=0.7, text opacity=1}, legend cell align={left}]
                    \addplot[name path=A,color=cb-red-orange] table [x=n,y=size_bl_dp,y error=size_bl_dp_error, col sep=comma]{\tHyperbolicILP};
                    \addlegendentry{\algoDP}
                    \addplot[name path=A,color=cb-pink] table [x=n,y=size_bl_dp_s,y error=size_bl_dp_s_error, col sep=comma]{\tHyperbolicILP};
                    \addlegendentry{\algoDPStar}
                    \addplot[name path=B,color=cb-blue] table [x=n,y=size_ilp, y error=size_ilp_error, col sep=comma]{\tHyperbolicILP};
                    \addlegendentry{ILP}
	       	    \addplot[name path=C,color=cb-magenta] table [x=n,y=size_bl_dp_pert_min_pc, y error=size_bl_dp_pert_min_pc_error, col sep=comma]{\tHyperbolicILP};
        	    \addlegendentry{\algoDPAll}
                    \addplot[name path=upper-blelloch,draw opacity=0.0] table[x=n,y expr=\thisrow{size_bl_dp}+\thisrow{size_bl_dp_error}] {\tHyperbolicILP};
                    \addplot[name path=lower-blelloch,draw opacity=0.0] table[x=n,y expr=\thisrow{size_bl_dp}-\thisrow{size_bl_dp_error}] {\tHyperbolicILP};
                    \addplot[color=cb-red-orange!20] fill between[of=upper-blelloch and lower-blelloch];

                    \addplot[name path=upper-ilp,draw opacity=0.0] table[x=n,y expr=\thisrow{size_ilp}+\thisrow{size_ilp_error}] {\tHyperbolicILP};
                    \addplot[name path=lower-ilp,draw opacity=0.0] table[x=n,y expr=\thisrow{size_ilp}-\thisrow{size_ilp_error}] {\tHyperbolicILP};
                    \addplot[color=cb-blue!20] fill between[of=upper-ilp and lower-ilp];

		    \addplot[name path=upper-ilp,draw opacity=0.0] table[x=n,y expr=\thisrow{size_bl_dp_pert_min_pc}+\thisrow{size_bl_dp_pert_min_pc_error}] {\tHyperbolicILP};
		    \addplot[name path=lower-ilp,draw opacity=0.0] table[x=n,y expr=\thisrow{size_bl_dp_pert_min_pc}-\thisrow{size_bl_dp_pert_min_pc_error}] {\tHyperbolicILP};
		    \addplot[color=cb-magenta!20] fill between[of=upper-ilp and lower-ilp];

		    \addplot[name path=upper-ilp,draw opacity=0.0] table[x=n,y expr=\thisrow{size_bl_dp_s}+\thisrow{size_bl_dp_s_error}] {\tHyperbolicILP};
		    \addplot[name path=lower-ilp,draw opacity=0.0] table[x=n,y expr=\thisrow{size_bl_dp_s}-\thisrow{size_bl_dp_s_error}] {\tHyperbolicILP};
		    \addplot[color=cb-pink!20] fill between[of=upper-ilp and lower-ilp];

                    \addplot[only marks,opacity=0] table[x=n, y=size_ilp] {\tHyperbolicILP}
                    \foreach \i in {0,...,\NodeRows} {
                            coordinate [pos=\i/\NodeRows] (a\i)
                        };
                    \addplot[only marks, opacity=0] table[x=n, y=size_bl_dp] {\tHyperbolicILP}
                    \foreach \i in {0,...,\NodeRows} {
                            coordinate [pos=\i/\NodeRows] (b\i)
                        };

        		\addplot[only marks, opacity=0] table[x=n, y=size_bl_dp_pert_min_pc] {\tHyperbolicILP}
        		\foreach \i in {0,...,\NodeRows} {
        			coordinate [pos=\i/\NodeRows] (c\i)
        		};

                    \pgfplotstablegetelem{\NodeRows}{size_ilp}\of\tHyperbolicILP
                    \pgfmathsetmacro{\ILPFinal}{\pgfplotsretval}

                    \pgfplotstablegetelem{\NodeRows}{size_bl_dp}\of\tHyperbolicILP
                    \pgfmathsetmacro{\BlellochFinal}{\pgfplotsretval}

                    \node (a) at (axis cs:70,\ILPFinal) {};
                    \node[right=0.15cm of a] (a-ex) {};
                    \node (b) at (axis cs:70,\BlellochFinal) {};
                    \node[right=0.15cm of b] (b-ex) {};
                    \pgfmathsetmacro{\ApproximationFactor}{\BlellochFinal/\ILPFinal}
                    \pgfkeys{/pgf/number format/.cd,fixed,precision=2}
                    \draw[dashed,color=black!50] (a-ex.center) -- node[near start,right=0.25cm,rotate=90,color=black!70] {\small$\sigma_b=\pgfmathprintnumber{\ApproximationFactor}$} (b-ex.center);
                    \draw[color=black!50] (a.center) -- (a-ex.center);
                    \draw[color=black!50] (b.center) -- (b-ex.center);

                    \pgfplotstablegetelem{0}{size_ilp}\of\tHyperbolicILP
                    \pgfmathsetmacro{\ILPStart}{\pgfplotsretval}

                    \pgfplotstablegetelem{0}{size_bl_dp}\of\tHyperbolicILP
                    \pgfmathsetmacro{\BlellochStart}{\pgfplotsretval}

                    \node (a) at (axis cs:20,\ILPStart) {};
                    \node[left=0.05cm of a] (a-ex) {};
                    \node (b) at (axis cs:20,\BlellochStart) {};
                    \node[left=0.05cm of b] (b-ex) {};
                    \pgfmathsetmacro{\ApproximationFactorSec}{\BlellochStart/\ILPStart}
                    \pgfkeys{/pgf/number format/.cd,fixed,precision=2}
                    \draw[dashed,color=black!50] (a-ex.center) -- node[midway,left=0.15cm,rotate=90,color=black!70,xshift=1cm] {\small$\sigma_b=\pgfmathprintnumber{\ApproximationFactorSec}$} (b-ex.center);
                    \draw[color=black!50] (a.center) -- (a-ex.center);
                    \draw[color=black!50] (b.center) -- (b-ex.center);
                \end{axis}
                \foreach \i in {0,\NodeRows} {
                        \pgfplotstablegetelem{\i}{size_ilp}\of\tHyperbolicILP
                        \pgfmathsetmacro{\ILPAvgSize}{\pgfplotsretval}
                        \pgfplotstablegetelem{\i}{size_bl_dp_pert_min_pc}\of\tHyperbolicILP
                        \pgfmathsetmacro{\BlellochAvgSize}{\pgfplotsretval}
                        \pgfmathsetmacro{\ApproximationFactor}{\BlellochAvgSize/\ILPAvgSize}
                        \pgfkeys{/pgf/number format/.cd,fixed,precision=2}

                        \ifthenelse{\equal{\i}{\NodeRows}}
			{\draw[dashed,color=black!50] (a\i) -- node[near start,right=0.18cm,rotate=90,color=black!70,xshift=-0.2cm] {\small$\sigma_h=\pgfmathprintnumber{\ApproximationFactor}$} (c\i);}
			{\draw[dashed,color=black!50] (a\i) -- node[right,midway,color=black!70,yshift=-0.2cm] {\small$\sigma_h=\pgfmathprintnumber{\ApproximationFactor}$} (c\i);}
                    }
            \end{tikzpicture}}
    \end{subfigure}
    \hfill
	\begin{subfigure}{0.32\textwidth}
		\scalebox{0.57}{
			\begin{tikzpicture}
				\pgfplotstablegetrowsof{\tILPPowerlaw}
				\pgfmathtruncatemacro\NodeRows{\pgfplotsretval-1}
				\begin{axis}[legend pos=north west, xmin=13, xmax=78,  xlabel=number of nodes,grid=both, title={$\mathcal{MC}(\cdot,100)$},legend style={fill opacity=0.7, text opacity=1}, legend cell align={left}]
					\addplot[name path=A,color=cb-red-orange] table [x=n,y=size_bl_dp,y error=size_bl_dp_error, col sep=comma]{\tILPPowerlaw};
					\addlegendentry{\algoDP}
					\addplot[name path=A,color=cb-pink] table [x=n,y=size_bl_dp_s,y error=size_bl_dp_s_error, col sep=comma]{\tILPPowerlaw};
					\addlegendentry{\algoDPStar}
					\addplot[name path=B,color=cb-blue] table [x=n,y=size_ilp, y error=size_ilp_error, col sep=comma]{\tILPPowerlaw};
					\addlegendentry{ILP}
					\addplot[name path=C,color=cb-magenta] table [x=n,y=size_bl_dp_pert_min_pc, y error=size_bl_dp_pert_min_pc_error, col sep=comma]{\tILPPowerlaw};
					\addlegendentry{\algoDPAll}
					\addplot[name path=upper-blelloch,draw opacity=0.0] table[x=n,y expr=\thisrow{size_bl_dp}+\thisrow{size_bl_dp_error}] {\tILPPowerlaw};
					\addplot[name path=lower-blelloch,draw opacity=0.0] table[x=n,y expr=\thisrow{size_bl_dp}-\thisrow{size_bl_dp_error}] {\tILPPowerlaw};
					\addplot[color=cb-red-orange!20] fill between[of=upper-blelloch and lower-blelloch];
					
					\addplot[name path=upper-ilp,draw opacity=0.0] table[x=n,y expr=\thisrow{size_ilp}+\thisrow{size_ilp_error}] {\tILPPowerlaw};
					\addplot[name path=lower-ilp,draw opacity=0.0] table[x=n,y expr=\thisrow{size_ilp}-\thisrow{size_ilp_error}] {\tILPPowerlaw};
					\addplot[color=cb-blue!20] fill between[of=upper-ilp and lower-ilp];

					\addplot[name path=upper-ilp,draw opacity=0.0] table[x=n,y expr=\thisrow{size_bl_dp_pert_min_pc}+\thisrow{size_bl_dp_pert_min_pc_error}] {\tILPPowerlaw};
					\addplot[name path=lower-ilp,draw opacity=0.0] table[x=n,y expr=\thisrow{size_bl_dp_pert_min_pc}-\thisrow{size_bl_dp_pert_min_pc_error}] {\tILPPowerlaw};
					\addplot[color=cb-magenta!20] fill between[of=upper-ilp and lower-ilp];

					\addplot[name path=upper-ilp,draw opacity=0.0] table[x=n,y expr=\thisrow{size_bl_dp_s}+\thisrow{size_bl_dp_s_error}] {\tILPPowerlaw};
					\addplot[name path=lower-ilp,draw opacity=0.0] table[x=n,y expr=\thisrow{size_bl_dp_s}-\thisrow{size_bl_dp_s_error}] {\tILPPowerlaw};
					\addplot[color=cb-pink!20] fill between[of=upper-ilp and lower-ilp];
					
					\addplot[only marks,opacity=0] table[x=n, y=size_ilp] {\tILPPowerlaw}
					\foreach \i in {0,...,\NodeRows} {
						coordinate [pos=\i/\NodeRows] (a\i)
					};
					\addplot[only marks, opacity=0] table[x=n, y=size_bl_dp] {\tILPPowerlaw}
					\foreach \i in {0,...,\NodeRows} {
						coordinate [pos=\i/\NodeRows] (b\i)
					};

					\addplot[only marks, opacity=0] table[x=n, y=size_bl_dp_pert_min_pc] {\tILPPowerlaw}
					\foreach \i in {0,...,\NodeRows} {
						coordinate [pos=\i/\NodeRows] (c\i)
					};

                                        \pgfplotstablegetelem{\NodeRows}{size_ilp}\of\tILPPowerlaw
                                        \pgfmathsetmacro{\ILPFinal}{\pgfplotsretval}

                                        \pgfplotstablegetelem{\NodeRows}{size_bl_dp}\of\tILPPowerlaw
                                        \pgfmathsetmacro{\BlellochFinal}{\pgfplotsretval}

                                        \node (a) at (axis cs:70,\ILPFinal) {};
                                        \node[right=0.15cm of a] (a-ex) {};
                                        \node (b) at (axis cs:70,\BlellochFinal) {};
                                        \node[right=0.15cm of b] (b-ex) {};
                                        \pgfmathsetmacro{\ApproximationFactor}{\BlellochFinal/\ILPFinal}
                                        \pgfkeys{/pgf/number format/.cd,fixed,precision=2}
                                        \draw[dashed,color=black!50] (a-ex.center) -- node[near start,right=0.25cm,rotate=90,color=black!70] {\small$\sigma_b=\pgfmathprintnumber{\ApproximationFactor}$} (b-ex.center);
                                        \draw[color=black!50] (a.center) -- (a-ex.center);
                                        \draw[color=black!50] (b.center) -- (b-ex.center);

                                        \pgfplotstablegetelem{0}{size_ilp}\of\tILPPowerlaw
                                        \pgfmathsetmacro{\ILPStart}{\pgfplotsretval}

                                        \pgfplotstablegetelem{0}{size_bl_dp}\of\tILPPowerlaw
                                        \pgfmathsetmacro{\BlellochStart}{\pgfplotsretval}

                                        \node (a) at (axis cs:20,\ILPStart) {};
                                        \node[left=0.05cm of a] (a-ex) {};
                                        \node (b) at (axis cs:20,\BlellochStart) {};
                                        \node[left=0.05cm of b] (b-ex) {};
                                        \pgfmathsetmacro{\ApproximationFactorSec}{\BlellochStart/\ILPStart}
                                        \pgfkeys{/pgf/number format/.cd,fixed,precision=2}
                                        \draw[dashed,color=black!50] (a-ex.center) -- node[midway,left=0.15cm,rotate=90,color=black!70,xshift=1cm] {\small$\sigma_b=\pgfmathprintnumber{\ApproximationFactorSec}$} (b-ex.center);
                                        \draw[color=black!50] (a.center) -- (a-ex.center);
                                        \draw[color=black!50] (b.center) -- (b-ex.center);
					
				\end{axis}
				\foreach \i in {0,\NodeRows} {
					\pgfplotstablegetelem{\i}{size_ilp}\of\tILPPowerlaw
					\pgfmathsetmacro{\ILPAvgSize}{\pgfplotsretval}
					\pgfplotstablegetelem{\i}{size_bl_dp_pert_min_pc}\of\tILPPowerlaw
					\pgfmathsetmacro{\BlellochAvgSize}{\pgfplotsretval}
					\pgfmathsetmacro{\ApproximationFactor}{\BlellochAvgSize/\ILPAvgSize}
					\pgfkeys{/pgf/number format/.cd,fixed,precision=2}
					
					\ifthenelse{\equal{\i}{\NodeRows}}
					{\draw[dashed,color=black!50] (a\i) -- node[near start,right=0.18cm,rotate=90,color=black!70,xshift=-0.2cm] {\small$\sigma_h=\pgfmathprintnumber{\ApproximationFactor}$} (c\i);}
					{\draw[dashed,color=black!50] (a\i) -- node[right,midway,color=black!70,yshift=-0.2cm] {\small$\sigma_h=\pgfmathprintnumber{\ApproximationFactor}$} (c\i);}
				}
		\end{tikzpicture}}
	\end{subfigure}
	\hfill
    \begin{subfigure}{0.32\textwidth}
        \scalebox{0.57}{
            \begin{tikzpicture}
                \pgfplotstablegetrowsof{\tILP}
                \pgfmathtruncatemacro\NodeRows{\pgfplotsretval-1}
                \begin{axis}[legend pos=north west, xmin=16, xmax=55, xlabel=number of nodes,grid=both, title={$\mathcal{G}(n,\ 3/n)$},legend style={fill opacity=0.7, text opacity=1}, legend cell align={left}]
                    \draw [fill=black, opacity=0.04] (axis cs:40,-100) rectangle (1000,1000);
                    \node (lb) at (axis cs:48,100) {\tiny less than 85\% solved};
                    \addplot[name path=A,color=cb-red-orange] table [x=n,y=size_bl_dp,y error=size_bl_dp_error, col sep=comma]{\tILP};
                    \addlegendentry{\algoDP}

                    \addplot[name path=A,color=cb-pink] table [x=n,y=size_bl_dp_s,y error=size_bl_dp_s_error, col sep=comma]{\tILP};
                    \addlegendentry{\algoDPStar}
                    \addplot[name path=B,color=cb-blue] table [x=n,y=size_ilp, y error=size_ilp_error, col sep=comma]{\tILP};
                    \addlegendentry{ILP}

    		\addplot[name path=C,color=cb-magenta] table [x=n,y=size_bl_dp_pert_min_pc, y error=size_bl_dp_pert_min_pc_error, col sep=comma]{\tILP};
    		\addlegendentry{\algoDPAll}

		\addplot[name path=upper-ilp,draw opacity=0.0] table[x=n,y expr=\thisrow{size_bl_dp_pert_min_pc}+\thisrow{size_bl_dp_pert_min_pc_error}] {\tILP};
		\addplot[name path=lower-ilp,draw opacity=0.0] table[x=n,y expr=\thisrow{size_bl_dp_pert_min_pc}-\thisrow{size_bl_dp_pert_min_pc_error}] {\tILP};
		\addplot[color=cb-magenta!20] fill between[of=upper-ilp and lower-ilp];

                \addplot[name path=upper-blelloch,draw opacity=0.0] table[x=n,y expr=\thisrow{size_bl_dp}+\thisrow{size_bl_dp_error}] {\tILP};
                    \addplot[name path=lower-blelloch,draw opacity=0.0] table[x=n,y expr=\thisrow{size_bl_dp}-\thisrow{size_bl_dp_error}] {\tILP};
                    \addplot[color=cb-red-orange!20] fill between[of=upper-blelloch and lower-blelloch];

                    \addplot[name path=upper-ilp,draw opacity=0.0] table[x=n,y expr=\thisrow{size_ilp}+\thisrow{size_ilp_error}] {\tILP};
                    \addplot[name path=lower-ilp,draw opacity=0.0] table[x=n,y expr=\thisrow{size_ilp}-\thisrow{size_ilp_error}] {\tILP};
                    \addplot[color=cb-blue!20] fill between[of=upper-ilp and lower-ilp];

                    \addplot[name path=upper-ilp,draw opacity=0.0] table[x=n,y expr=\thisrow{size_bl_dp_s}+\thisrow{size_bl_dp_s_error}] {\tILP};
                    \addplot[name path=lower-ilp,draw opacity=0.0] table[x=n,y expr=\thisrow{size_bl_dp_s}-\thisrow{size_bl_dp_s_error}] {\tILP};
                    \addplot[color=cb-pink!20] fill between[of=upper-ilp and lower-ilp];

                    \addplot[only marks,opacity=0] table[x=n, y=size_bl_dp] {\tILP}
                    \foreach \i in {0,...,\NodeRows} {
                            coordinate [pos=\i/\NodeRows] (a\i)
                        };
                    \addplot[only marks, opacity=0] table[x=n, y=size_ilp] {\tILP}
                    \foreach \i in {0,...,\NodeRows} {
                            coordinate [pos=\i/\NodeRows] (b\i)
                        };

                    \addplot[only marks, opacity=0] table[x=n, y=size_bl_dp_pert_min_pc] {\tILP}
                    \foreach \i in {0,...,\NodeRows} {
                            coordinate [pos=\i/\NodeRows] (c\i)
                        };

                    \pgfplotstablegetelem{\NodeRows}{size_ilp}\of\tILP
                    \pgfmathsetmacro{\ILPFinal}{\pgfplotsretval}

                    \pgfplotstablegetelem{\NodeRows}{size_bl_dp}\of\tILP
                    \pgfmathsetmacro{\BlellochFinal}{\pgfplotsretval}

                    \node (a) at (axis cs:50,\ILPFinal) {};
                    \node[right=0.15cm of a] (a-ex) {};
                    \node (b) at (axis cs:50,\BlellochFinal) {};
                    \node[right=0.15cm of b] (b-ex) {};
                    \pgfmathsetmacro{\ApproximationFactor}{\BlellochFinal/\ILPFinal}
                    \pgfkeys{/pgf/number format/.cd,fixed,precision=2}
                    \draw[dashed,color=black!50] (a-ex.center) -- node[near start,right=0.25cm,rotate=90,color=black!70] {\small$\sigma_b=\pgfmathprintnumber{\ApproximationFactor}$} (b-ex.center);
                    \draw[color=black!50] (a.center) -- (a-ex.center);
                    \draw[color=black!50] (b.center) -- (b-ex.center);

                    \pgfplotstablegetelem{0}{size_ilp}\of\tILP
                    \pgfmathsetmacro{\ILPStart}{\pgfplotsretval}

                    \pgfplotstablegetelem{0}{size_bl_dp}\of\tILP
                    \pgfmathsetmacro{\BlellochStart}{\pgfplotsretval}

                    \node (a) at (axis cs:20,\ILPStart) {};
                    \node[left=0.05cm of a] (a-ex) {};
                    \node (b) at (axis cs:20,\BlellochStart) {};
                    \node[left=0.05cm of b] (b-ex) {};
                    \pgfmathsetmacro{\ApproximationFactorSec}{\BlellochStart/\ILPStart}
                    \pgfkeys{/pgf/number format/.cd,fixed,precision=2}
                    \draw[dashed,color=black!50] (a-ex.center) -- node[midway,left=0.18cm,rotate=90,color=black!70,xshift=1cm] {\small$\sigma_b=\pgfmathprintnumber{\ApproximationFactorSec}$} (b-ex.center);
                    \draw[color=black!50] (a.center) -- (a-ex.center);
                    \draw[color=black!50] (b.center) -- (b-ex.center);
                \end{axis}

                \foreach \i in {0,\NodeRows} {
                        \pgfplotstablegetelem{\i}{size_ilp}\of\tILP
                        \pgfmathsetmacro{\ILPAvgSize}{\pgfplotsretval}
                        \pgfplotstablegetelem{\i}{size_bl_dp_pert_min_pc}\of\tILP
                        \pgfmathsetmacro{\BlellochAvgSize}{\pgfplotsretval}
                        \pgfmathsetmacro{\ApproximationFactor}{\BlellochAvgSize/\ILPAvgSize}
                        \pgfkeys{/pgf/number format/.cd,fixed,precision=2}

                        \ifthenelse{\equal{\i}{\NodeRows}}
                        {\draw[dashed,color=black!50] (b\i) -- node[near start,right=0.18cm,color=black!70,rotate=90] {\small$\sigma_h=\pgfmathprintnumber{\ApproximationFactor}$} (c\i);}
                        {\draw[dashed,color=black!50] (b\i) -- node[right,midway,color=black!70,yshift=-0.2cm] {\small$\sigma_h=\pgfmathprintnumber{\ApproximationFactor}$} (c\i);}
                    }
            \end{tikzpicture}}
    \end{subfigure}

        \begin{subfigure}[t]{0.98\textwidth}
            \centering
        \scalebox{0.57}{
            \begin{tikzpicture} 
        \begin{axis}[
        hide axis,
        xmin=10,
        xmax=50,
        ymin=0,
        ymax=0.4,
        legend style={at={(0.5,-0.1)},draw=black!50,anchor=north,legend columns=-1, align=left},
        ]
        \addlegendimage{/pgfplots/ybar,cb-red-orange,fill=cb-red-orange!20}
        \addlegendentry{\algoDP ~~~}

        \addlegendimage{/pgfplots/ybar,cb-magenta,fill=cb-magenta!20}
        \addlegendentry{\algoDPAll ~~~}

        \addlegendimage{line legend, smooth, cb-pink, dashed}
        \addlegendentry{$f(n)=0.42n$}
        \end{axis}
    \end{tikzpicture}}
        \end{subfigure}

    \begin{subfigure}{0.34\textwidth}
        \scalebox{0.57}{
            \begin{tikzpicture}
                        \pgfkeys{/pgf/number format/.cd,fixed,precision=1}
                \begin{axis}[ybar=3pt, bar width=0.2cm, ylabel={approximation factor $\sigma_{max}$},xmax=73,xmin=17, xlabel={number of nodes $n$},grid style={solid,black!15}, ymajorgrids, legend pos=north west,legend style={fill opacity=0.7, text opacity=1}, legend cell align={left}]
                    \addplot[color=cb-red-orange, fill=cb-red-orange!20, nodes near coords] table [x=n,y=approximation_factor_blelloch,col sep=comma] {\tHyperbolicMaxILP};
        		\addplot[color=cb-magenta, fill=cb-magenta!20, nodes near coords, every node/.style = {rotate=90, xshift=0.3cm, yshift=-0.4cm}] table [x=n,y=approximation_factor_bl_dp_pert_min_pc,col sep=comma] {\tHyperbolicMaxILP};
                    \addplot[smooth, no markers, domain=15:75, color=cb-pink, dashed, forget plot] {0.42*x};
                \end{axis}
                        \pgfkeys{/pgf/number format/.cd,fixed,precision=2}
            \end{tikzpicture}}
    \end{subfigure}
	\hfill
	\begin{subfigure}{0.32\textwidth}
		\scalebox{0.57}{
			\begin{tikzpicture}
                        \pgfkeys{/pgf/number format/.cd,fixed,precision=1}
    			\begin{axis}[ybar=3pt, bar width=0.2cm, xmax=73,xmin=16, xlabel={number of nodes $n$},grid style={solid,black!15}, ymajorgrids, legend pos=north west,legend style={fill opacity=0.7, text opacity=1}]
        			\addplot[color=cb-red-orange, fill=cb-red-orange!20, nodes near coords] table [x=n,y=approximation_factor_blelloch,col sep=comma] {\tILPMaxPowerlaw};
        			\addplot[color=cb-magenta, fill=cb-magenta!20, nodes near coords, every node/.style = {rotate=90, xshift=0.3cm, yshift=-0.4cm}] table [x=n,y=approximation_factor_bl_dp_pert_min_pc,col sep=comma] {\tILPMaxPowerlaw};
        			\addplot[smooth, no markers, domain=15:75, color=cb-pink, dashed, forget plot] {0.42*x};
    			\end{axis}
                        \pgfkeys{/pgf/number format/.cd,fixed,precision=2}
		\end{tikzpicture}}
	\end{subfigure}
    \hfill
    \begin{subfigure}{0.32\textwidth}
        \scalebox{0.57}{
            \begin{tikzpicture}
                \begin{axis}[ybar, bar width=0.2cm, xmax=53,xmin=17, xlabel={number of nodes $n$},grid style={solid,black!15}, ymajorgrids, legend pos=north west,legend style={fill opacity=0.7, text opacity=1}, legend cell align={left}]
                    \addplot[color=cb-red-orange, fill=cb-red-orange!20, nodes near coords] table [x=n,y=approximation_factor_blelloch,col sep=comma] {\tILPMax};
        		\addplot[color=cb-magenta, fill=cb-magenta!20, nodes near coords, every node/.style = {rotate=90, xshift=0.3cm, yshift=-0.4cm}] table [x=n,y=approximation_factor_bl_dp_pert_min_pc,col sep=comma] {\tILPMax};
                \end{axis}
            \end{tikzpicture}}
    \end{subfigure}
    \caption{The comparison between the optimal algorithm and the heuristics for the \probMSP problem for $k=2,\ \rho=n-1$ on several random graph classes for varying $n$. $\sigma_h$ and $\sigma_b$ describe the average approximation factor for \algoDPAll and \algoDP respectively, while $\sigma_{max} = \max_i\sigma_i$ is the maximum approximation factor wittnessed on up to 50 sampled instances.}
    \label{fig:k-2-approximation-factor}
\end{figure}
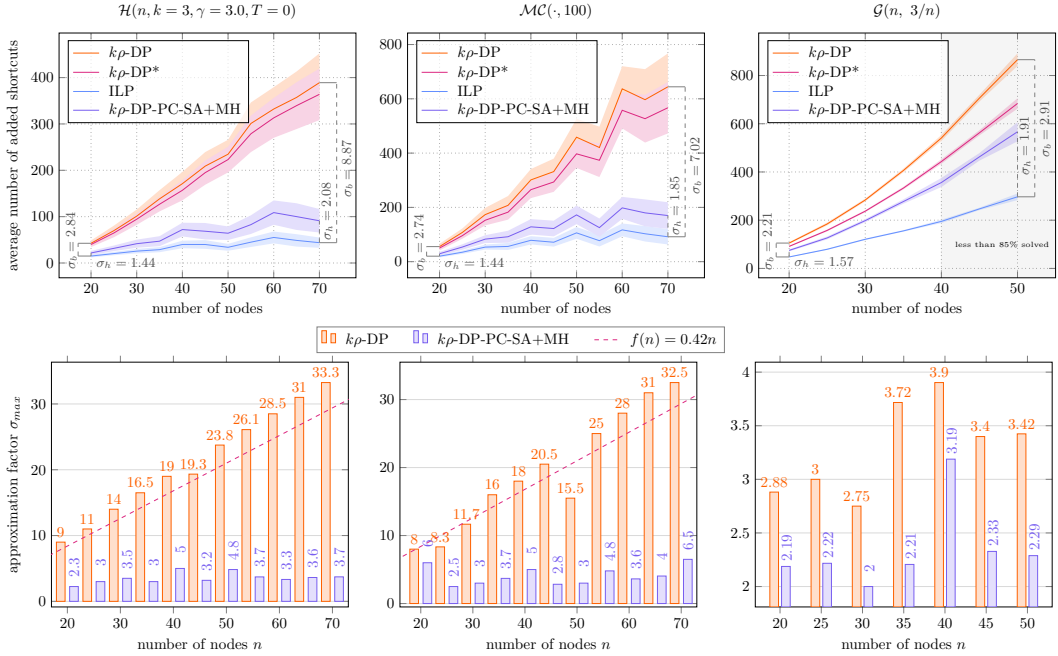
\begin{figure}[t]
    \begin{subfigure}[t]{0.33\textwidth}
    \scalebox{0.52}{
    \begin{tikzpicture}
        \centering
        \begin{axis}[tick scale binop=\times,clip marker paths=true,ylabel={average solving time of solved instances (s)},xlabel={number of nodes}, grid=both, legend pos=north west,error bars/y dir=both, error bars/y explicit, title={Average ILP solving time}, legend cell align={left}]
            \addplot[color=cb-pink,mark size=1.5pt,mark=o] table[x=n, y=speed_ilp_s,y error=speed_ilp_s_error] {\tILP};
            \addlegendentry{$\mathcal{G}(n,p)$}

            \addplot[color=cb-blue,mark size=1.5pt,mark=square] table[x=n, y=speed_ilp_s,y error=speed_ilp_s_error] {\tILPPowerlaw};
            \addlegendentry{$\mathcal{MC}(\cdot,s)$}

            \addplot[color=cb-orange,mark size=1.5pt,mark=diamond] table[x=n, y=speed_ilp_s,y error=speed_ilp_s_error] {\tHyperbolicILP};
            \addlegendentry{$\mathcal{H}(n,k,\gamma,T)$}
        \end{axis}
    \end{tikzpicture}}
    \caption{The average time for the ILP solver to solve the random graph instances with respect to their order.}
    \label{fig:misc-time}
    \end{subfigure}
    \hfill
    \begin{subfigure}[t]{0.43\textwidth}
    \scalebox{0.52}{
    \begin{tikzpicture}
        \centering
        \begin{axis}[width=12cm,height=7.25cm,tick scale binop=\times,clip marker paths=true,ylabel={factor of edge blowup},xlabel={number of nodes}, grid=both, legend pos=north east,error bars/y dir=both, error bars/y explicit,error bars/error bar style={solid}, title={Average edge blowup},legend style={fill opacity=0.7, text opacity=1, column sep=0.15cm}, legend cell align={left}]
            \addplot[color=cb-pink,mark size=1.5pt,mark=o] table[x=n, y=m_blowup_ilp,y error=m_blowup_ilp_error] {\tILP};
            \addlegendentry{ILP $\mathcal{G}(n,p)$}
            \addplot[color=cb-pink,dashed,mark options={solid}, mark size=1.5pt,mark=o] table[x=n, y=m_blowup_blelloch,y error=m_blowup_blelloch_error] {\tILP};
            \addlegendentry{\algoDP $\mathcal{G}(n,p)$}

            \addplot[color=cb-orange,mark size=1.5pt,mark=o] table[x=n, y=m_blowup_ilp,y error=m_blowup_ilp_error] {\tHyperbolicILP};
            \addlegendentry{ILP $\mathcal{H}(n,k,\gamma,T)$}
            \addplot[color=cb-orange,dashed,mark options={solid}, mark size=1.5pt,mark=o] table[x=n, y=m_blowup_blelloch,y error=m_blowup_blelloch_error] {\tHyperbolicILP};
            \addlegendentry{\algoDP $\mathcal{H}(n,k,\gamma,T)$}

            \addplot[color=cb-blue,mark size=1.5pt,mark=o] table[x=n, y=m_blowup_ilp,y error=m_blowup_ilp_error] {\tILPPowerlaw};
            \addlegendentry{ILP $\mathcal{MC}(\cdot,s)$}
            \addplot[color=cb-blue,dashed,mark options={solid}, mark size=1.5pt,mark=o] table[x=n, y=m_blowup_blelloch,y error=m_blowup_blelloch_error] {\tILPPowerlaw};
            \addlegendentry{\algoDP $\mathcal{MC}(\cdot,s)$}
        \end{axis}
    \end{tikzpicture}}
    \caption{The average edge blowup \ie the ratio of edges in the shortcutted graph to the edges in the input graph~${|E\cup S|/|E|}$ where $S$ is the generated shortcut set.}
    \label{fig:edge-blowup}
    \end{subfigure}
    \hfill
    \begin{subfigure}[t]{0.21\textwidth}
    \scalebox{0.52}{
    \begin{tikzpicture}
        \centering
        \begin{axis}[height=7.2cm,width=4.8cm,name=thirty,tick scale binop=\times,clip marker paths=true,ylabel={solving time (s)}, title={$\mathcal{G}(n,3/n)$}, xlabel={approximation factor $\sigma$}, grid=both, legend cell align={left}]
            \addplot[scatter, only marks, point meta min=-700, point meta max=1000] table[x=approximation_factor_blelloch, y=speed_ilp_s] {\tILPRawthirtyfive};
        \end{axis}
        \node[xshift=-1.2cm,yshift=-0.5cm] (a) at (thirty.north east) {86\% solved};
    \end{tikzpicture}}
    \caption{Solving time as a function of the approximation factor $\sigma$ for $n=35$.}
    \label{fig:solving-time}
    \end{subfigure}
    \caption{Miscellaneous results on including average running time, average factor of edge increases and solavbility on Gilbert random graphs.}
    \label{fig:misc}
\end{figure}
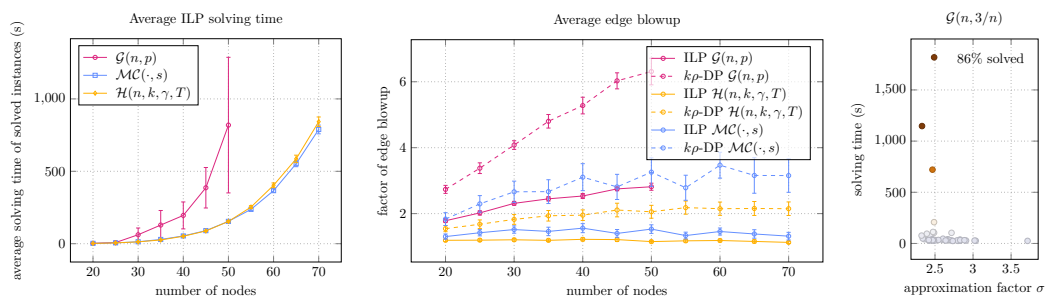

\begin{figure}
    \begin{subfigure}[t]{0.98\textwidth}
        \centering
    \scalebox{0.55}{
        \begin{tikzpicture} 
    \begin{axis}[
    hide axis,
    xmin=10,
    xmax=50,
    ymin=0,
    ymax=0.4,
    legend style={at={(0.5,-0.1)},draw=black!50,anchor=north,legend columns=-1, align=left},
    ]
    \addlegendimage{/pgfplots/ybar,cb-orange,fill=cb-orange!20}
    \addlegendentry{\algoGreedy ~~~}

    \addlegendimage{/pgfplots/ybar,cb-pink,fill=cb-pink!20}
    \addlegendentry{\algoDP ~~~}

    \addlegendimage{cb-red-orange,fill=cb-red-orange!20,/pgfplots/ybar}
    \addlegendentry{\algoDPStar ~~~}

    \addlegendimage{cb-blue,fill=cb-blue!20,/pgfplots/ybar}
    \addlegendentry{\algoDPPCMH ~~~}

    \addlegendimage{cb-magenta,fill=cb-magenta!20,/pgfplots/ybar}
    \addlegendentry{\algoDPSAMH ~~~}

    \addlegendimage{line legend, smooth, black, dashed}
    \addlegendentry{\algoDPAll}
    \end{axis}
\end{tikzpicture}}
    \end{subfigure}

    \begin{subfigure}[t]{0.32\textwidth}
    \scalebox{0.53}{
    \begin{tikzpicture}
        \centering
        \begin{axis}[ybar=2pt,xmin=0.4,xmax=6.6,bar width=0.14cm, xtick={1,2,3,4,5,6}, xticklabels={$\sqrt{n}$,$n/10$,$n/4$,$n/3$,$n/2$,$n-1$}, ylabel={average approximation factor $\sigma$},xlabel={$\rho$}, grid=both, legend pos=north east,error bars/y dir=both, error bars/y explicit, title={$\mathcal{H}(n,k=3,\gamma=3.0,T=0)$},legend style={fill opacity=0.7, text opacity=1},ymode=log, log basis y={2}]
            \addplot[color=cb-orange, fill=cb-orange!20] table[x=regime, y=approx_greedy, y error=approx_greedy_error] {\tLargeComparisonPowerlawThree};

            \addplot[color=cb-pink, fill=cb-pink!20] table[x=regime, y=approx, y error=approx_error] {\tLargeComparisonHyperbolicThree};

            \addplot[color=cb-red-orange, fill=cb-red-orange!20] table[x=regime, y=approx_ms, y error=approx_ms_error] {\tLargeComparisonHyperbolicThree};

            \addplot[color=cb-blue, fill=cb-blue!20] table[x=regime, y=approx_pc, y error=approx_pc_error] {\tLargeComparisonHyperbolicThree};

            \addplot[color=cb-magenta, fill=cb-magenta!20] table[x=regime, y=approx_pert, y error=approx_pert_error] {\tLargeComparisonHyperbolicThree};

            \addplot[smooth, no markers, domain=-1:6, color=black,dashed,forget plot]{1};
        \end{axis}
    \end{tikzpicture}}
    \end{subfigure}
    \hfill
    \begin{subfigure}[t]{0.32\textwidth}
    \scalebox{0.53}{
    \begin{tikzpicture}
        \centering
        \begin{axis}[ybar=2pt,xmin=0.4,xmax=6.6,bar width=0.14cm, xtick={1,2,3,4,5,6}, xticklabels={$\sqrt{n}$,$n/10$,$n/4$,$n/3$,$n/2$,$n-1$}, ylabel={average approximation factor $\sigma$},xlabel={$\rho$}, grid=both, legend pos=north east,error bars/y dir=both, error bars/y explicit, title={$\mathcal{MC}(\cdot,100)$},legend style={fill opacity=0.7, text opacity=1},ymode=log, log basis y={2}]
            \addplot[color=cb-orange, fill=cb-orange!20] table[x=regime, y=approx_greedy, y error=approx_greedy_error] {\tLargeComparisonPowerlawThree};

            \addplot[color=cb-pink, fill=cb-pink!20] table[x=regime, y=approx, y error=approx_error] {\tLargeComparisonPowerlawThree};

            \addplot[color=cb-red-orange, fill=cb-red-orange!20] table[x=regime, y=approx_ms, y error=approx_ms_error] {\tLargeComparisonPowerlawThree};

            \addplot[color=cb-blue, fill=cb-blue!20] table[x=regime, y=approx_pc, y error=approx_pc_error] {\tLargeComparisonPowerlawThree};

            \addplot[color=cb-magenta, fill=cb-magenta!20] table[x=regime, y=approx_pert, y error=approx_pert_error] {\tLargeComparisonPowerlawThree};

            \addplot[smooth, no markers, domain=-1:6, color=black,dashed,forget plot]{1};
        \end{axis}
    \end{tikzpicture}}
    \end{subfigure}
    \hfill
    \begin{subfigure}[t]{0.32\textwidth}
    \scalebox{0.53}{
    \begin{tikzpicture}
        \centering
        \begin{axis}[ybar=2pt,xmin=0.4,xmax=6.6,bar width=0.14cm, xtick={1,2,3,4,5,6}, xticklabels={$\sqrt{n}$,$n/10$,$n/4$,$n/3$,$n/2$,$n-1$}, ylabel={average approximation factor $\sigma$},xlabel={$\rho$}, grid=both, legend pos=north east,error bars/y dir=both, error bars/y explicit, title={$\mathcal{G}(n,3/n)$},legend style={fill opacity=0.7, text opacity=1},ymode=log, log basis y={2}]
            \addplot[color=cb-orange, fill=cb-orange!20] table[x=regime, y=approx_greedy, y error=approx_greedy_error] {\tLargeComparisonGilbertThree};

            \addplot[color=cb-pink, fill=cb-pink!20] table[x=regime, y=approx, y error=approx_error] {\tLargeComparisonGilbertThree};

            \addplot[color=cb-red-orange, fill=cb-red-orange!20] table[x=regime, y=approx_ms, y error=approx_ms_error] {\tLargeComparisonGilbertThree};

            \addplot[color=cb-blue, fill=cb-blue!20] table[x=regime, y=approx_pc, y error=approx_pc_error] {\tLargeComparisonGilbertThree};

            \addplot[color=cb-magenta, fill=cb-magenta!20] table[x=regime, y=approx_pert, y error=approx_pert_error] {\tLargeComparisonGilbertThree};

            \addplot[smooth, no markers, domain=-1:6, color=black,dashed,forget plot]{1};
        \end{axis}
    \end{tikzpicture}}
    \end{subfigure}
    \caption{The average approximation factor of all heuristics compared to the baseline \algoDPAll on various random graph classes for $n \leq 8000$ (the number of nodes fluctuates with the size of the largest connected component) and $k=3$.}
    \label{fig:k-3-heuristics}
\end{figure}
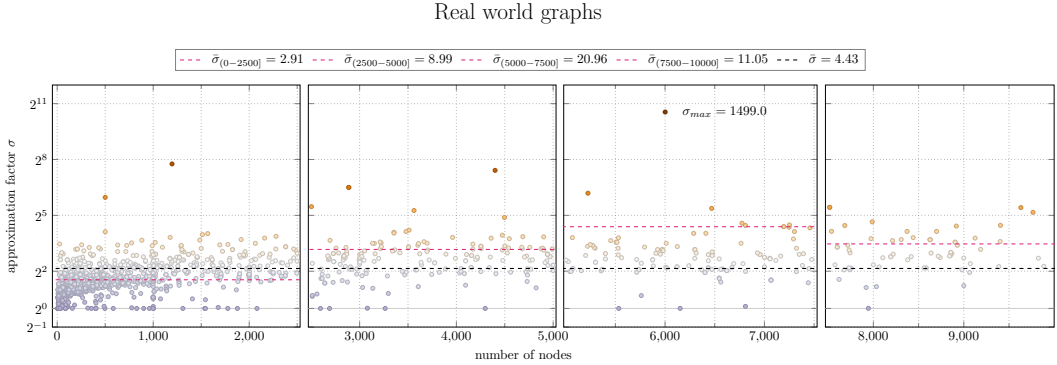
\begin{figure}[t]
    \begin{subfigure}{0.98\textwidth}
        \centering
        \scalebox{0.5}{\LARGE Real world graphs}
    \end{subfigure}\\[1.5ex]
    \begin{subfigure}[t]{0.98\textwidth}
        \centering
    \scalebox{0.5}{
        \begin{tikzpicture} 
    \begin{axis}[
    hide axis,
    xmin=10,
    xmax=50,
    ymin=0,
    ymax=0.4,
    legend style={at={(0.5,-0.1)},draw=black!50,anchor=north,legend columns=-1, align=left},
    ]
    \pgfplotstablegetelem{0}{approx}\of\tRealWorldDataProcessedFirstMean
    \pgfmathsetmacro\mac{\pgfplotsretval}
    \addlegendimage{dashed,cb-pink,thick}
    \addlegendentry{$\bar \sigma_{(0-2500]}=\pgfmathprintnumber{\mac}$}

    \pgfplotstablegetelem{0}{approx}\of\tRealWorldDataProcessedSecondMean
    \pgfmathsetmacro\mact{\pgfplotsretval}
    \addlegendimage{dashed,cb-pink,thick}
    \addlegendentry{$\bar \sigma_{(2500-5000]}=\pgfmathprintnumber{\mact}$}

    \pgfplotstablegetelem{0}{approx}\of\tRealWorldDataProcessedThirdMean
    \pgfmathsetmacro\macr{\pgfplotsretval}
    \addlegendimage{dashed,cb-pink,thick}
    \addlegendentry{$\bar \sigma_{(5000-7500]}=\pgfmathprintnumber{\macr}$}

    \pgfplotstablegetelem{0}{approx}\of\tRealWorldDataProcessedFourthMean
    \pgfmathsetmacro\macf{\pgfplotsretval}
    \addlegendimage{dashed,cb-pink,thick}
    \addlegendentry{$\bar \sigma_{(7500-10000]}=\pgfmathprintnumber{\macf}$}

    \pgfplotstablegetelem{0}{approx}\of\tRealWorldDataProcessedMean
    \pgfmathsetmacro\macg{\pgfplotsretval}
    \addlegendimage{dashed,black,thick}
    \addlegendentry{$\bar \sigma=\pgfmathprintnumber{\macg}$}
    \end{axis}
\end{tikzpicture}}
    \end{subfigure}\\[1ex]
    \begin{subfigure}{0.27\textwidth}
        \centering
    \scalebox{0.5}{
    \begin{tikzpicture}
        \begin{axis}[ymode=log,height=8cm,width=8.1cm, log basis y={2}, xmin=-50,xmax=2530, ymin=0.5,ymax=4096,xtick={0,1000,2000,3000},ytick={0.5,1,4,32,256,2048}, grid=both,clip marker paths=true, ylabel={approximation factor $\sigma$}, legend pos=north west,minor tick num=1, extra y ticks={1},extra y tick style={grid style={solid,gray!40}}, extra y tick labels=\empty,colormap name={PuOr-11}]
            \addplot[scatter, only marks,mark size=1.5pt,point meta=explicit,point meta min={-4},point meta max={9},forget plot] table[x=n,y=approx,meta expr=log2(\thisrow{approx})]{\tRealWorldDataProcessedFirst};
            \pgfplotstablegetelem{0}{approx}\of\tRealWorldDataProcessedFirstMean
            \pgfmathsetmacro\mac{\pgfplotsretval}
            \addplot[smooth, no marks,domain=-100:2600,dashed,color=cb-pink,thick] {\mac};

            \pgfplotstablegetelem{0}{approx}\of\tRealWorldDataProcessedMean
            \pgfmathsetmacro\macg{\pgfplotsretval}
            \addplot[smooth, no marks,domain=-100:2600,dashed,color=black,thick] {\macg};
        \end{axis}
    \end{tikzpicture}}
    \end{subfigure}
    \hfill
    \begin{subfigure}{0.23\textwidth}
        \centering
    \scalebox{0.5}{
    \begin{tikzpicture}
        \begin{axis}[ymode=log,height=8cm,width=8.1cm, log basis y={2}, xmin=2470,xmax=5030, ymin=0.5,ymax=4096,xtick={3000,4000,5000},ytick={0.5,4,32,256,2048}, yticklabels={\empty}, grid=both,clip marker paths=true, legend pos=north west,minor tick num=1, extra y ticks={1},extra y tick style={grid style={solid,gray!40}}, extra y tick labels=\empty,colormap name={PuOr-11}]
            \addplot[scatter, only marks,mark size=1.5pt,point meta=explicit,point meta min={-4},point meta max={9},forget plot] table[x=n,y=approx,meta expr=log2(\thisrow{approx})]{\tRealWorldDataProcessedSecond};
            \pgfplotstablegetelem{0}{approx}\of\tRealWorldDataProcessedSecondMean
            \pgfmathsetmacro\mac{\pgfplotsretval}
            \addplot[smooth, no marks,domain=2400:5100,dashed,color=cb-pink,thick] {\mac};

            \pgfplotstablegetelem{0}{approx}\of\tRealWorldDataProcessedMean
            \pgfmathsetmacro\macg{\pgfplotsretval}
            \addplot[smooth, no marks,domain=2400:5100,dashed,color=black,thick] {\macg};
        \end{axis}
    \end{tikzpicture}}
    \end{subfigure}
    \hfill
    \begin{subfigure}{0.23\textwidth}
        \centering
    \scalebox{0.5}{
    \begin{tikzpicture}
        \begin{axis}[ymode=log,height=8cm,width=8.25cm, log basis y={2}, xmin=4980,xmax=7530, ymin=0.5,ymax=4096,xtick={4979,6000,7000,8000},ytick={0.5,4,32,256,2048}, yticklabels={\empty}, grid=both,clip marker paths=true, legend pos=north west, minor xtick={5500,6500,7500}, extra y ticks={1},extra y tick style={grid style={solid,gray!40}}, extra y tick labels=\empty,colormap name={PuOr-11}]
            \addplot[scatter, only marks,mark size=1.5pt,point meta=explicit,point meta min={-4},point meta max={9},forget plot] table[x=n,y=approx,meta expr=log2(\thisrow{approx})]{\tRealWorldDataProcessedThird};
            \pgfplotstablegetelem{0}{approx}\of\tRealWorldDataProcessedThirdMean
            \pgfmathsetmacro\mac{\pgfplotsretval}
            \addplot[smooth, no marks,domain=4800:7600,dashed,color=cb-pink,thick] {\mac};
            \node (label) at (axis cs:6600, 1499) {$\sigma_{max}=1499.0$};

            \pgfplotstablegetelem{0}{approx}\of\tRealWorldDataProcessedMean
            \pgfmathsetmacro\macg{\pgfplotsretval}
            \addplot[smooth, no marks,domain=4800:7600,dashed,color=black,thick] {\macg};
        \end{axis}
    \end{tikzpicture}}
    \end{subfigure}
    \hfill
    \begin{subfigure}{0.24\textwidth}
        \centering
    \scalebox{0.5}{
    \begin{tikzpicture}
        \begin{axis}[ymode=log,height=8cm,width=7.6cm, log basis y={2}, xmin=7470,xmax=9999, ymin=0.5,ymax=4096,xtick={8000,9000,10000},ytick={0.5,4,32,256,2048}, yticklabels={\empty}, extra y tick labels=\empty, extra y tick style={grid style={solid,gray!40}}, grid=both,clip marker paths=true, legend pos=north west, minor tick num=1, extra y ticks={1} ,colormap name={PuOr-11}]
            \addplot[scatter, only marks,mark size=1.5pt,point meta=explicit,point meta min={-4},point meta max={9},forget plot] table[x=n,y=approx,meta expr=log2(\thisrow{approx})]{\tRealWorldDataProcessedFourth};
            \pgfplotstablegetelem{0}{approx}\of\tRealWorldDataProcessedFourthMean
            \pgfmathsetmacro\mac{\pgfplotsretval}
            \addplot[smooth, no marks,domain=7400:10100,dashed,color=cb-pink,thick] {\mac};

            \pgfplotstablegetelem{0}{approx}\of\tRealWorldDataProcessedMean
            \pgfmathsetmacro\macg{\pgfplotsretval}
            \addplot[smooth, no marks,domain=7400:10100,dashed,color=black,thick] {\macg};
        \end{axis}
    \end{tikzpicture}}
    \end{subfigure}\\[-2ex]
    \begin{subfigure}{0.99\textwidth}
        \centering
        \scalebox{0.5}{number of nodes}
    \end{subfigure}
    \caption{Approximation factor of \algoDP in relation to the baseline \algoDPAll on all graphs from the network repository dataset~\cite{nr} with at most $10000$ nodes with $\rho=n-1$ and $k=3$.}
    \label{fig:real-world-data}
\end{figure}
Our construction in \cref{sec:lower_bounds} to bound the approximation factor of \algoDP relies on carefully crafted graph structures.
This raises the question, how the algorithm performs on ``ordinary'' networks.
Thus our first experiment considers the approximation factor of \algoDP on $\mathcal{G}(n,p)$ graphs.  and two other random graph models.
For each model and graph size, we sample $50$ graphs and derive the average and maximum approximation factor.

For the $\mathcal{G}(n,p)$ model, \algoDP displays a small approximation factor (less than three) in any case which is slowly growing with increasing graph order as depicted in \cref{fig:k-2-approximation-factor}.
In addition, the approximation factor is strongly concentrated around the average, suggesting that there are no commonly found exceptionally hard instances for \algoDP on  $\mathcal{G}(n,p)$.

Notice that our lower bound construction shares some similarities with graphs demonstrating a powerlaw degree distribution, namely the existence of high degree vertices which, if shortcutted just right, may complete many $(k,\rho)$-balls at once.
Hence we conjecture that these graphs are a hard input for \algoDP.
Indeed our experiments show a quickly increasing average approximation factor for growing $n$ for both random graph models with powerlaw degree distributions.
For these two models, the maximum approximation factor for each data parameter point scales almost linearly in $n$.
Especially the result on the $\mathcal{MC}(\cdot,s)$ model (\ie uniform samples of simple graphs with a given degree sequence) suggest that the underlying cause thereof is likely the powerlaw degree distribution.

\algoDPAll combines the techniques described in \cref{subsec:pc,subsec:sa} with the additional usage of MinHashing.
Clearly, our heuristic performs better on average than \algoDP on every tested random graph model as depicted in \cref{fig:k-2-approximation-factor}.
This is especially prominent for graphs with powerlaw distributed vertex degrees. 
For the largest graphs of order $n=70$, on average our heuristic provides solutions which are $4.26$ and $3.79$ times smaller, for hyperbolic and $\mathcal{MC}(\cdot,s)$ respectively, than \algoDP's.
The effect is even larger on the maximum observed approximation factor, where our heuristic consistently yields smaller solutions --- sometimes improving one order in magnitude.
Surprisingly, we are able to match and improve upon the good performance of \algoDP  on Gilbert  graphs.

Observe that our heuristic is a multi-stage algorithm as it first invokes \algoDPPC, then \algoDPSA and finally \algoDP (\ie each stage observes previously inserted shortcuts).
This alone provides our heuristic with a considerable advantage compared to \algoDP which works completely oblivious of the other computed shortcuts.
Thus let us define \algoDPStar, an algorithm which repeatedly calls \algoDP on a fraction of $\frac{n}{10}$ nodes, inserts all proposed shortcuts and repeats this process until all nodes have calculated their shortcut set.
Evidently this simulates a continuously growing globally known shortcut set and allows us to further discern the root of better performance of our heuristic.
As reported in \cref{fig:k-2-approximation-factor}, \algoDPStar barely improves over \algoDP on graphs with powerlaw degree distributions, while it performs significantly better on $\mathcal{G}(n,p)$. 
We attribute these observation to the fact, that multiple stages do not help to find and place globally good shortcuts as needed to solve powerlaw degree instance; yet they uncover ``accidental'' synergies in the large number of shortcuts needed for Gilbert graphs.

\subsubsection{Measuring run time}
As depicted in \cref{fig:misc-time}, the average run time of our ILP solver grows dramatically in the graph size~$n$ in particular for Gilbert graphs.
While we can solve all instances in the reported parameter regime for all other models within the time budget of 1800\,s, the fraction of solved instances on Gilbert graphs shrinks with increasing~$n$ as detailed in \cref{table:solvability}.
Thus we visually highlight the regime with only partial results in \cref{fig:k-2-approximation-factor} to clearly mark the potentially biased results.
Despite various preliminary tests, no association between the solving time and the approximation factor could be found in this regime as depicted in \cref{fig:solving-time}.
Hence this reduces the probability that the unsolved instances induce a systemic bias on the investigated measures due to an association to the solving time.

For an input graph~$G=(V, E)$ and a computed shortcut set~$S$ define the \emph{edge blowup factor} as the relative increase $|E\cup S|/|E|$ of edges.
Recall that a main use case for $(k,\rho)$-graphs are sharper theoretical bounds for \probPSSSP algorithms.
The blowup factor directly affects the work of these algorithms, at least multiplicatively.

Surprisingly, for random graphs with powerlaw degree distributions, the ILP solution displays a small edge blowup factor of approximately $1.4$ which in addition seems to decrease for an increasing graph sizes.
This suggests that $(k, \rho)$ transformations are, on average, asymptotically irrelevant on these practically relevant graph classes.
In comparison \algoDP initially sets out with an edge blowup factor of less than two for both aforementioned models which in line with previous results, subsequently displays an increasing trend on graphs of higher order.
Contrasting this are the results for Gilbert random graphs, here both the heuristic and the ILP algorithm persistently increase their edge blowup factor for larger graphs.

Observe that the number of reachable nodes $\rho$ can be chosen freely within $[1,n)$.
In order to generalize our experiments, we investigate the performance of our heuristic on larger graphs for different values of $\rho$.
Recall that for $\rho=o(\sqrt{n})$ the shortcutting algorithm has a lower span than all \probPSSSP algorithms on $(k,\rho)$-graphs.
Hence we focus on $\rho=\Omega(\sqrt{n})$ for which our heuristic has the same span bounds as \algoDP.
We generate graphs with $n\in\set{2000,4000,6000,8000}$ nodes, extract the largest connected component, and consider the resulting (in general smaller) graph.
Since these graphs are out of reach for our ILP formulation, we define \algoDPAll as the baseline against which all other heuristics compete.
In addition, we explicitly state the performance of the separately described preprocessing steps in \cref{sec:preprocessing} to provide further insights into their individual impacts on the solution quality.
As depicted in \cref{fig:k-3-heuristics}, our heuristic consistently outperforms \algoDP, although the performance gap narrows quickly with decreasing $\rho$.
This is not surprising, as our heuristic benefits from shortcut synergies which become rarer with smaller $\rho$.
Consistently with the findings of Blelloch \etal, \algoDP's greedy counterpart performs worse than any other tested heuristic. 
When restricting our attention to the differential performance of \algoDPSAMH and \algoDPPCMH, we notice several interesting results: \begin{inparaenum}
    \item \algoDPPC shows only a superior performance in comparison to \algoDPSAMH for very high values of $\rho$. For smaller values of $\rho$, its performance quickly degrades until it is eventually overtaken by \algoDPSAMH. This trend is reproducible for varying choices of $k$; see \cref{fig:k-other-approx}.
    \item The final solution size of \algoDPAll is lower than the minimum of both other heuristics on their own. 
\end{inparaenum}
The previous discussion highlights the synergistic relationship of \algoDPPCMH and \algoDPSAMH; they excel in different parameter regimes of $\rho$ and their combination consistently outperforms their individual performance (although see \cref{fig:k-other-approx} for a narrow counter-example).

\subsubsection{Performance on real world graphs}
To quantify the robustness of our new heuristic, we also consider observed networks from \cite{nr}, including protein interactions, social, citation and neurological networks.
We test the heuristics on 2350 graphs from \cite{nr} (those with less than 10000 nodes which are not already $(k,\rho)$-graphs).
The results are split into four separate plots where every one displays a specific range of $n$.
We again use \algoDPAll as baseline as illustrated in \cref{fig:real-world-data}.
On the investigated graphs \algoDPAll either matches the performance of \algoDP or improves upon it by up to several orders of magnitude.
This culminates into a maximum observed approximation factor of $1499.0$, \ie on average for each edge added by \algoDPAll, \algoDP adds $1499$.
Even without exceptional instances, the average approximation factor over all tested graphs is $\bar \sigma=4.43$, which seems to increase on average for larger $n$.
This provides strong empirical evidence that our new heuristic is robust on practical instances.

\section{Conclusion}
We showed new results regarding the complexity of \probMSP and gave important theoretical insights into the lower bounds of the approximation factor of existing heuristics.
These results allowed us to derive a new heuristic extending upon the work of Blelloch \etal which experimentally shows a good performance not only in random graph models but also in a wide variety of real world graphs.
Finally, we were able to pinpoint a property commonly found in social network graphs, a powerlaw degree distribution, as one of the major factors inducing a large approximation factor for the previously existing heuristics.
Our results therefore uncovered new theoretical insights but were also able to supply a practical contribution.
Several open questions remain:
\begin{enumerate}
    \item As previously described, \algoDPSA and \algoDPPC excel at different parameter spaces for $\rho$, where exactly are the boundaries of these spaces?
    \item Are the constants negligible such that the new heuristics translate into efficient implementations?
    \item Is there an approximation algorithm for \probMSP with a sub-linear approximation factor?
\begin{longversion}
        \item Can one extend the hardness results for \probMSP to $k=2$. This would settle the question from which point on the problem gets intractable as it is known that for $k=1$ the problem is solvable in polynomial time.
        \item Let $f(n,k)$ be a monotonically decreasing function. Is there such a $f(n,k)$ such that the cardinality of the shortcut set is bounded above by $f(n,k)$? For $k=1$ it is known that the cardinality of the shortcut set is trivially bounded by $f(n,1) = \Oh(n^2)$.
\end{longversion}
\end{enumerate}

\bibliography{main}

\appendix
\section{Appendix}
\subsection{Supporting figures and tables}
\begin{table}[H]
    \centering
    \begin{tabular}{|c|ccccccc|}
        \hline
        \thead{\textbf{n}} & 20& 25& 30& 35& 40& 45& 50\\\hline
        \thead{\textbf{solved fraction of instances}} & 1.0 & 1.0 & 1.0 & 0.86 & 0.66 & 0.54 & 0.3125\\\hline
    \end{tabular}
    \caption{The fraction of solved instances of the ILP solver within 1800s for the $\mathcal{G}(n,p)$ random graph model.}
    \label{table:solvability}
\end{table}
\begin{figure}[H]
    \begin{subfigure}[t]{0.98\textwidth}
        \centering
    \scalebox{0.53}{
        \begin{tikzpicture} 
    \begin{axis}[
    hide axis,
    xmin=10,
    xmax=50,
    ymin=0,
    ymax=0.4,
    legend style={at={(0.5,-0.1)},draw=black!50,
    anchor=north,legend columns=-1, align=left},
    ]
    \addlegendimage{/pgfplots/ybar,cb-orange,fill=cb-orange!20}
    \addlegendentry{\algoGreedy ~~~}

    \addlegendimage{/pgfplots/ybar,cb-pink,fill=cb-pink!20}
    \addlegendentry{\algoDP ~~~}

    \addlegendimage{cb-red-orange,fill=cb-red-orange!20,/pgfplots/ybar}
    \addlegendentry{\algoDPStar ~~~}

    \addlegendimage{cb-blue,fill=cb-blue!20,/pgfplots/ybar}
    \addlegendentry{\algoDPPCMH ~~~}

    \addlegendimage{cb-magenta,fill=cb-magenta!20,/pgfplots/ybar}
    \addlegendentry{\algoDPSAMH ~~~}

    \addlegendimage{line legend, smooth, black, dashed}
    \addlegendentry{\algoDPAll}
    \end{axis}
\end{tikzpicture}}
    \end{subfigure}

    \begin{subfigure}[t]{0.34\textwidth}
    \scalebox{0.53}{
    \begin{tikzpicture}
        \centering
        \begin{axis}[ybar=2pt,xmin=0.4,xmax=6.6,bar width=0.14cm, xtick={1,2,3,4,5,6}, xticklabels={$\sqrt{n}$,$n/10$,$n/4$,$n/3$,$n/2$,$n-1$}, ylabel style={align=center,execute at begin node=\setlength{\baselineskip}{2em}},ylabel={{\LARGE$k=4$}\\average approximation factor $\sigma$},xlabel={$\rho$}, grid=both, legend pos=north east,error bars/y dir=both, error bars/y explicit, title={$\mathcal{H}(n,k=3,\gamma=3.0,T=0)$},legend style={fill opacity=0.7, text opacity=1},ymode=log,log basis y={2}]
            \addplot[color=cb-orange, fill=cb-orange!20] table[x=regime, y=approx_greedy, y error=approx_greedy_error] {\tLargeComparisonHyperbolicFour};
            \addplot[color=cb-pink, fill=cb-pink!20] table[x=regime, y=approx, y error=approx_error] {\tLargeComparisonHyperbolicFour};

            \addplot[color=cb-red-orange, fill=cb-red-orange!20] table[x=regime, y=approx_ms, y error=approx_ms_error] {\tLargeComparisonHyperbolicFour};

            \addplot[color=cb-blue, fill=cb-blue!20] table[x=regime, y=approx_pc, y error=approx_pc_error] {\tLargeComparisonHyperbolicFour};

            \addplot[color=cb-magenta, fill=cb-magenta!20] table[x=regime, y=approx_pert, y error=approx_pert_error] {\tLargeComparisonHyperbolicFour};

            \addplot[smooth, no markers, domain=-1:6, color=black,dashed,forget plot]{1};
        \end{axis}
    \end{tikzpicture}}
    \end{subfigure}
    \hfill
    \begin{subfigure}[t]{0.31\textwidth}
    \scalebox{0.53}{
    \begin{tikzpicture}
        \centering
        \begin{axis}[ybar=2pt,xmin=0.4,xmax=6.6,bar width=0.14cm, xtick={1,2,3,4,5,6}, xticklabels={$\sqrt{n}$,$n/10$,$n/4$,$n/3$,$n/2$,$n-1$},xlabel={$\rho$}, grid=both, legend pos=north east,error bars/y dir=both, error bars/y explicit, title={$\mathcal{MC}(\cdot,100)$},legend style={fill opacity=0.7, text opacity=1},ymode=log, log basis y={2}]
            \addplot[color=cb-orange, fill=cb-orange!20] table[x=regime, y=approx_greedy, y error=approx_greedy_error] {\tLargeComparisonPowerlawFour};
            \addplot[color=cb-pink, fill=cb-pink!20] table[x=regime, y=approx, y error=approx_error] {\tLargeComparisonPowerlawFour};

            \addplot[color=cb-red-orange, fill=cb-red-orange!20] table[x=regime, y=approx_ms, y error=approx_ms_error] {\tLargeComparisonPowerlawFour};

            \addplot[color=cb-blue, fill=cb-blue!20] table[x=regime, y=approx_pc, y error=approx_pc_error] {\tLargeComparisonPowerlawFour};

            \addplot[color=cb-magenta, fill=cb-magenta!20] table[x=regime, y=approx_pert, y error=approx_pert_error] {\tLargeComparisonPowerlawFour};

            \addplot[smooth, no markers, domain=-1:6, color=black,dashed,forget plot]{1};
        \end{axis}
    \end{tikzpicture}}
    \end{subfigure}
    \hfill
    \begin{subfigure}[t]{0.31\textwidth}
    \scalebox{0.53}{
    \begin{tikzpicture}
        \centering
        \begin{axis}[ybar=2pt,xmin=0.4,xmax=6.6,bar width=0.14cm, xtick={1,2,3,4,5,6}, xticklabels={$\sqrt{n}$,$n/10$,$n/4$,$n/3$,$n/2$,$n-1$},xlabel={$\rho$}, grid=both, legend pos=north east,error bars/y dir=both, error bars/y explicit, title={$\mathcal{G}(n,3/n)$},legend style={fill opacity=0.7, text opacity=1},ymode=log, log basis y={2}]
            \addplot[color=cb-orange, fill=cb-orange!20] table[x=regime, y=approx_greedy, y error=approx_greedy_error] {\tLargeComparisonGilbertFour};
            \addplot[color=cb-pink, fill=cb-pink!20] table[x=regime, y=approx, y error=approx_error] {\tLargeComparisonGilbertFour};

            \addplot[color=cb-red-orange, fill=cb-red-orange!20] table[x=regime, y=approx_ms, y error=approx_ms_error] {\tLargeComparisonGilbertFour};

            \addplot[color=cb-blue, fill=cb-blue!20] table[x=regime, y=approx_pc, y error=approx_pc_error] {\tLargeComparisonGilbertFour};

            \addplot[color=cb-magenta, fill=cb-magenta!20] table[x=regime, y=approx_pert, y error=approx_pert_error] {\tLargeComparisonGilbertFour};

            \addplot[smooth, no markers, domain=-1:6, color=black,dashed,forget plot]{1};
        \end{axis}
    \end{tikzpicture}}
    \end{subfigure}

    \begin{subfigure}[t]{0.34\textwidth}
    \scalebox{0.53}{
    \begin{tikzpicture}
        \centering
        \begin{axis}[ybar=2pt,xmin=0.4,xmax=6.6,bar width=0.14cm, xtick={1,2,3,4,5,6}, xticklabels={$\sqrt{n}$,$n/10$,$n/4$,$n/3$,$n/2$,$n-1$},ylabel style={align=center,execute at begin node=\setlength{\baselineskip}{2em}}, ylabel={{\LARGE$k=5$}\\average approximation factor $\sigma$},xlabel={$\rho$}, grid=both, legend pos=north east,error bars/y dir=both, error bars/y explicit,legend style={fill opacity=0.7, text opacity=1},ymode=log,log basis y={2}]
            \addplot[color=cb-orange, fill=cb-orange!20] table[x=regime, y=approx_greedy, y error=approx_greedy_error] {\tLargeComparisonHyperbolicFive};
            \addplot[color=cb-pink, fill=cb-pink!20] table[x=regime, y=approx, y error=approx_error] {\tLargeComparisonHyperbolicFive};

            \addplot[color=cb-red-orange, fill=cb-red-orange!20] table[x=regime, y=approx_ms, y error=approx_ms_error] {\tLargeComparisonHyperbolicFive};

            \addplot[color=cb-blue, fill=cb-blue!20] table[x=regime, y=approx_pc, y error=approx_pc_error] {\tLargeComparisonHyperbolicFive};

            \addplot[color=cb-magenta, fill=cb-magenta!20] table[x=regime, y=approx_pert, y error=approx_pert_error] {\tLargeComparisonHyperbolicFive};

            \addplot[smooth, no markers, domain=-1:6, color=black,dashed,forget plot]{1};
        \end{axis}
    \end{tikzpicture}}
    \end{subfigure}
    \hfill
    \begin{subfigure}[t]{0.31\textwidth}
    \scalebox{0.53}{
    \begin{tikzpicture}
        \centering
        \begin{axis}[ybar=2pt,xmin=0.4,xmax=6.6,bar width=0.14cm, xtick={1,2,3,4,5,6}, xticklabels={$\sqrt{n}$,$n/10$,$n/4$,$n/3$,$n/2$,$n-1$},xlabel={$\rho$}, grid=both, legend pos=north east,error bars/y dir=both, error bars/y explicit,legend style={fill opacity=0.7, text opacity=1},ymode=log, log basis y={2}]
            \addplot[color=cb-orange, fill=cb-orange!20] table[x=regime, y=approx_greedy, y error=approx_greedy_error] {\tLargeComparisonPowerlawFive};
            \addplot[color=cb-pink, fill=cb-pink!20] table[x=regime, y=approx, y error=approx_error] {\tLargeComparisonPowerlawFive};

            \addplot[color=cb-red-orange, fill=cb-red-orange!20] table[x=regime, y=approx_ms, y error=approx_ms_error] {\tLargeComparisonPowerlawFive};

            \addplot[color=cb-blue, fill=cb-blue!20] table[x=regime, y=approx_pc, y error=approx_pc_error] {\tLargeComparisonPowerlawFive};

            \addplot[color=cb-magenta, fill=cb-magenta!20] table[x=regime, y=approx_pert, y error=approx_pert_error] {\tLargeComparisonPowerlawFive};

            \addplot[smooth, no markers, domain=-1:6, color=black,dashed,forget plot]{1};
        \end{axis}
    \end{tikzpicture}}
    \end{subfigure}
    \hfill
    \begin{subfigure}[t]{0.32\textwidth}
    \scalebox{0.53}{
    \begin{tikzpicture}
        \centering
        \begin{axis}[ybar=2pt,xmin=0.4,xmax=6.6,bar width=0.14cm, xtick={1,2,3,4,5,6}, xticklabels={$\sqrt{n}$,$n/10$,$n/4$,$n/3$,$n/2$,$n-1$},xlabel={$\rho$}, grid=both, legend pos=north east,error bars/y dir=both, error bars/y explicit,legend style={fill opacity=0.7, text opacity=1},ymode=log, log basis y={2}]
            \addplot[color=cb-orange, fill=cb-orange!20] table[x=regime, y=approx_greedy, y error=approx_greedy_error] {\tLargeComparisonGilbertFive};

            \addplot[color=cb-pink, fill=cb-pink!20] table[x=regime, y=approx, y error=approx_error] {\tLargeComparisonGilbertFive};

            \addplot[color=cb-red-orange, fill=cb-red-orange!20] table[x=regime, y=approx_ms, y error=approx_ms_error] {\tLargeComparisonGilbertFive};

            \addplot[color=cb-blue, fill=cb-blue!20] table[x=regime, y=approx_pc, y error=approx_pc_error] {\tLargeComparisonGilbertFive};

            \addplot[color=cb-magenta, fill=cb-magenta!20] table[x=regime, y=approx_pert, y error=approx_pert_error] {\tLargeComparisonGilbertFive};

            \addplot[smooth, no markers, domain=-1:6, color=black,dashed,forget plot]{1};
        \end{axis}
    \end{tikzpicture}}
    \end{subfigure}
    \caption{Shows the average approximation factor of all heuristics compared to the baseline \algoDPAll on various random graph classes for $n \leq 8000$ (the number of nodes fluctuates with the size of the largest connected component).}
    \label{fig:k-other-approx}
\end{figure}

\begin{longversion}
    \subsection{Omitted proof for Lemma~\ref{lemma:single_ball_for_complex}}
\label{sec:append_complex_shortcuts}
In the following, we establish a number of properties of the transformations of \cref{subsec:vc_transform}.
They accumulate to a formal proof of \cref{lemma:single_ball_for_complex}, which is restated at the end of the section.
Let $G = (V, E)$ be the input graph for \probVC, $\Gtrans = (\Vtrans, \Etrans)$ the graph obtained from the transformation in \cref{subsec:vc_transform}, and $\Gop = (\Vtrans, \Etrans \cup S)$ where $S$ is a minimal cardinality shortcut set to ensure $\Gop$ is a $(k, \rho)$ graph.
Let us denote by $T_w$ a shortest path tree emanating from $w$.
Since $G$ and $\Gtrans$ are undirected graphs with unit edge weights, any shortest path tree is a shortest path tree with fewest hops.

\begin{lemma} \label{lemma:at-least-one}
    Any shortcut in an optimal solution set $S$ for \probMSP forms at least one $(k,\rho)$-ball for some node $v \in \Vedges$ in $\Gop$.
\end{lemma}
\begin{proof}
    For any $s \in S$, the graph $\Gop = (\Vtrans,\Etrans\cup (S\setminus \{s\}))$ is not a $(k,\rho)$-graph.
    Otherwise, the solution was not minimal and as such not optimal to begin with.
    Hence, $s$ forms at least one $(k,\rho)$-ball for some node $v \in \Vedges$.
\end{proof}

\begin{lemma} \label{lemma:tree}
    Given an undirected graph $G_n=(V_n,E_n)$ with strictly positive edge weights and a shortcut $s=\{u,w\}$.
    Any node $v \in V$ decreases the hop distance to at least one other node iff the nodes connected by $s$ fulfill $\hat d(v,u) < \hat d(v,w)-1$ in $G_n$ and $u$ is an ancestor of $w$ in a shortest path tree $T_v$ rooted in $v$.
\end{lemma}
\begin{proof}
    Fix any node $v$ which decreases its hop distance to at least one other node by the addition of $s$ to $G_n$.
    We now assume without loss of generality there is no shortest path tree with fewest hops for which $u$ is an ancestor of $w$.
    Then, necessarily the shortcut creates a new path which is not a shortest path, since $d(v,u)+d(u,w) > d(v,w)$.
    Thus, the shortcut cannot decrease the hop distance to any node.
    Hence, without loss of generality we assume $u$ is an ancestor of $w$ and such a shortest path tree exists.
    This implies by the assumption of strictly positive edge weights that $d(v,u) < d(v,w)$ and that $v\leadsto w \leadsto u$ is not a shortest path from $v$ to $w$.
    It directly follows that in order to decrease the hop distance from $v$ to $w$ the shortcut endpoint $u$ has to be closer by number of hops to $v$ before the insertion of the shortcut, thereby implying $\hat d(v,u) < \hat d(v,w)-1$ in $G_n$.
    Then, the hop distance from $v\leadsto w$ is decreased and consequently for all other nodes $x$ inside the subtree rooted by $w$ in $T_v$ for which $v\leadsto w\leadsto x$ is a shortest path with fewest hops too.
\end{proof}
\noindent
It is clear that our constructed graph conforms to the assumptions made by the previous \cref{lemma:tree} for any number of shortcuts added.

\begin{lemma} \label{lemma:shortcut-more-than-zero}
    Any shortcut $s$ in an optimal solution set for \probMSP has both endpoints within $k$ hop distance of any node $v \in \Vedges$ in the graph $\Go$ for which the addition of $s$ increases the number of $\rho$ nearest neighbors reachable in $k$ hops by more than $1$.
\end{lemma}
\begin{proof}
    Observe that shortcuts cannot lie arbitrary in the graph $\Go$.
    By construction any node $v\in \Vedges$ has its $\rho$ nearest neighbors within $k+1$ hop distance.
    Thus, the shortcuts are required to enable a path with fewer hops to the nodes within $k+1$ hop distance.
    By \cref{lemma:at-least-one} we fixate an arbitrary vertex $w$ which forms a $(k,\rho)$-ball through the addition of the shortcut.
    Notice that in order for $w$ to benefit from the shortcut both endpoints have to be within $k+1$ hop distance of $w$ in $\Go$.
    Let $s=\{u,v\}$ be the shortcut in question, by \cref{lemma:tree} there exists a shortest path tree $T_w$ and one endpoint is closer in terms of hops to $w$.
    Fix the vertex that has a larger hop distance and without loss of generality call it $u$.
    Any shortcut decreases the hop distance of all nodes in the subtree rooted in $u$ of $T_w$.
    Therefore, any shortcut to a node with hop distance $k+1$ will only decrease the hop distance to that node and all other direct neighbors of that node are already at hop distance $k+2$, in consequence not valid for forming a $(k,\rho)$-ball.
    This does not hold for arbitrary graphs, but it does for the constructed one, since we chose unit edge weights, all nodes at hop distance $k+1$ in $\Go$ have the same distance.
    Consequently, a shortcut into a node with $k+1$ hop distance can never construct a shortest path to other nodes with $k+1$ hop distance.
    Hence, any shortcut in an optimal solution which increases the number of reachable neighbors by more than one has to end on a node with hop distance at most $k$ and \cref{lemma:tree} imposes that the other endpoint is then at most at hop distance $k-2$ from $w$ in $\Go$.
\end{proof}

\begin{lemma} \label{lemma:dist-two-nodes}
    The distance of two nodes from $\Vedges$ in $\Go$ is always a multiple of $2\cdot(k-2)$ and the distance of two nodes from $\Vedges$ for which the original edges were not adjacent in the input graph is at least $4k-8$.
\end{lemma}
\begin{proof}
    Let $G'$ be the graph obtained from using the first step of \cref{subsec:vc_transform} for $k=3$. Then:
    \begin{equation*}
        \hat d_{G'}(u,v) = \frac{\hat d_{\Gtrans}(u,v)}{k-2} \quad \forall u,v \in \Vedges
    \end{equation*}
    Since $\hat d_{G'}(u,v)$ is exactly $2$ for edge nodes whose respective edges were adjacent in the original graph the claim follows.
\end{proof}

\begin{lemma} \label{lemma:optimal-solution-size}
    Let $v\in \Vedges$ be an edge vertex and $U\subseteq S$ be the subset of shortcuts in an optimal solution set $S$ for \probMSP, that that adds only one node to the $\rho$ nearest neighbors reached in $k$ hops to $v$. Then, $|U| \leq |V|-1|$ and $v$ still forms a $(k,\rho')$-ball in $\Gop=(\Vtrans,\Etrans\cup (S\setminus U))$, where $\rho' = \rho-(|V|-1) = 6|V|+7$.
\end{lemma}
\begin{proof}
    An implication of \cref{lemma:shortcut} is that there is an optimal solution of size $|V|$.
    Hence, $|U| \leq |S| \leq |V|$ such that $\rho-|V| \geq 6\cdot |V|+6$, as $\rho$ was chosen sufficiently large.
    Therefore, $v$ does not form a $(k,\rho)$-ball through the addition of $U$ to $\Gtrans$.
    Thus, $|U| \leq |V|-1 < |S|$, as at least one shortcut has to include the remaining $2|V|+3$ (or 6|V|+7) nodes to form the $(k,\rho)$-ball for $v$.
    It follows, that after removing $U$ from $S$, the number of reachable nodes in $k$ hops for $v$ decreases by at most $\left(|V|-1)\right)$.
    Leading to $\rho' \geq \rho-(|V|-1) \geq 6|V|+7$.
\end{proof}

\begin{remark} \label{remark:shortcut-more-than-const}
    Since we cannot accurately estimate the impact of shortcuts contributing one or less nodes to $(k,\rho)$-balls, we assume that the number of these shortcuts is maximal for every edge node.
    Hence, \cref{lemma:optimal-solution-size} allows us to ignore these shortcuts and still rule out specific types of other shortcuts as long as no $(k,\rho')$-graph is formable under these assumptions.
\end{remark}

\begin{lemma} \label{lemma:single-short-cor}
    Given an optimal solution set $S$ for \probMSP, then for every node $w \in \Vedges$ there exists exactly one shortcut $s \in S$ such that $w \in \Vedges$ forms a $(k,\rho)$-ball in $\Gop=(\Vtrans,\Etrans \cup \{s\})$ for $k\geq 3$ and $\gamma$ sufficiently large.
\end{lemma}
\begin{proof}
    By contradiction.
    First assume there is no subset of shortcuts $Y \subseteq S$ such that $w$ forms a $(k,\rho)$-ball, then $S$ was not a solution.
    Hence, assume that there is a subset $Y$ such that $w$ forms a $(k,\rho)$-ball.
    If only $w$ forms a $(k,\rho)$-ball through the subset of shortcuts $Y$ and $|Y|>1$ the solution is suboptimal since by \cref{lemma:shortcut} we can do better.
    Thus, assume $Y$ benefits at least one other node $x$ apart from $w$.
    Then by \cref{lemma:shortcut-more-than-zero} and \cref{remark:shortcut-more-than-const} both endpoints of the shortcut lie within $k$ hop distance of both $x$ and $w$ in $\Go$.
    Since both nodes simultaneously benefit from at least one shortcut it holds that $\exists s \in Y$ with $s=\{u,v\}$ such that by \cref{lemma:tree} both nodes $x$ and $w$ have hop distance $k$ and $k-2$ to one of the nodes in the shortcut.
    Therefore, the hop distance between $x$ and $w$ in the graph $\Go$ is bounded by
    \begin{equation*}
        \hat d_{\Go}(w,x) \leq \min(\hat d_{\Go}(w,u)+\hat d_{\Go}(u,x),\hat d_{\Go}(w,v)+\hat d_{\Go}(v,x)) \leq 2k-2
    \end{equation*}
    For any $k\geq 4$ this yields by \cref{lemma:dist-two-nodes} that the two nodes have a common node $h \in V$ such that shortcutting $h$'s $\gamma$-\pitchfork is better than using multiple shortcuts.
    Thus, the solution was not optimal.
    However, this analysis breaks for $k=3$.

    So assume now $k=3$.
    Notice that we don't have a problem if say $u$ is for both nodes $w$ and $x$ at hop distance $k-2$ then we would get $2k-4$ hop distance in $\Go$ which allows us to use the same argument as before.
    So the only case we need to consider is if both $u$ is at distance $k$ for $w$ and at distance $k-2$ for $x$ and the opposite for $v$.
    Then, we know that we connect a node at hop distance $1$ from $w$ or $x$ to a node at hop distance $3$ from $w$ or $x$.
    Observe that any nodes at hop distance $1$ and $3$ from any edge node $q \in \Vedges$ is a node $a \in V$ (for hop distance $1$ and $3$), hence an original node from the input graph or a node from the $\gamma$-\pitchfork (for hop distance $3$).
    So assume we connect two nodes from the original input graph by a shortcut (invalid for $k>3$).
    This increases the number of reachable nodes by at most $deg(G)+1 < |V|+1$ for all nodes $q \in \Vedges$ adjacent to either $u$ or $v$.
    As described by \cref{remark:shortcut-more-than-const}, we assume our solution to \probMSP only needs to form a $(k,\rho')$-graph and by that we can ignore the influence of all shortcuts only marginally (by $1$) improving the number of reachable nodes in $k$ hops.
    Thus, to overcome the threshold of $\rho' = 6|V|+7$ we would need at least $3$ shortcuts for every node in $a \in V$ such that $4|V|+4+\max(2\cdot (k-5),0)+3\cdot(|V|+1) = 7|V|+7 > \rho'' > 4|V|+4+\max(2\cdot (k-5),0)+2\cdot(|V|+1)$.
    Assuming one wants to shortcut all nodes in the subset $A \subseteq V$, then one would need $\frac{3|A|}{2}$ shortcuts while directly shortcutting the $\gamma$-\pitchforks only needs $|A|$ shortcuts.
    Therefore, this kind of shortcut will not produce a solution of minimal size.

    The final and remaining case connects some node $u$ from the original graph to the base $v$ of some $\gamma$-\pitchfork.
    If the $\gamma$-\pitchfork of $v$ is attached to $u$, then this implies that both $w$ and $x$ have an edge to $u$, thus this shortcut alone forms the $(k,\rho)$-ball for both $w$ and $x$.
    So either $|Y|=1$ or the initial solution was suboptimal.

    So assume now that $u$ is connected by a shortcut to the base $v$ of the $\gamma$-\pitchfork of some other node from the initial graph.
    By this, the node connected to $u$ (either $w$ or $x$) forms a $(k,\rho)$-ball through this shortcut.
    Without loss of generality call the remaining node not connected to $u$, $x$.
    Then from the perspective of $x$ a walk of the form $(l,\ldots,v,\ldots,l,\ldots,u)$ is shortcutted, where $l$ is the host node in $V$ of the $\gamma$-\pitchfork from which $v$ is the base.
    Thus, there is not shortest path tree for $x$ where $v$ is an ancestor of $u$, by \cref{lemma:tree} this shortcut does not benefit $x$ at all, a contradiction to the assumption.
    Hence, for all cases $|Y|=1$ such that any node $q\in \Vedges$ forms a $(k,\rho)$-ball by the addition of exactly one shortcut.
\end{proof}

\noindent
The previous \cref{lemma:single-short-cor} establishes that any optimal \probMSP solution uses a single shortcut to form a $(k,\rho)$-ball for any node instead of using a subset of shortcuts which only combined build the $(k,\rho)$-ball for some nodes.
This allows us to only consider optimal placements of single shortcuts as for any node $q \in \Vedges$ there is no relevant interaction between multiple shortcuts.

\begin{lemma} \label{lemma:single-solution}
    Any shortcut $s$ in an optimal solution set $S$ for \probMSP forms at least one $(k,\rho)$-ball for some node $v \in \Vedges$ in $\Gop = (\Vtrans, \Etrans\cup \{s\})$.
\end{lemma}
\begin{proof}
    Assume the contrary, $s$ does not form at least one $(k,\rho)$-ball.
    By \cref{lemma:single-short-cor}, all $(k,\rho)$-balls are formed by a single shortcut.
    Hence, the solution set $S\setminus \{s\}$ would still produce a solution and was therefore not optimal to begin with.
\end{proof}

\begin{remark}
    Notice that \cref{lemma:at-least-one} is strictly different from \cref{lemma:single-solution}.
    The first one allows us to peel a single shortcut of all shortcuts and fix the affected vertex while the second one uses the independence of shortcuts established by \cref{lemma:single-short-cor} to proof that any shortcut forms at least one $(k,\rho)$-ball independently of all other shortcuts.
\end{remark}

\singleBallForComplex*
\begin{proof}
    Fix two nodes edge nodes $u,v \in \Vedges$.
    We again split the proof for $k\geq 4$ and $k=3$, assume for now $k\geq 4$ and that $u$ and $v$ share a common neighbor $w \in V$, otherwise the proof of \cref{lemma:single-short-cor} states that they cannot interact in any meaningful way with each other.
    \cref{fig:shortest-path-tree} depicts the relevant case distinction for $k\geq 4$, for \complex shortcuts.
    In the following we use the notation \textit{(a)$\leadsto$(b)} to say we consider a shortcut connecting some node in \textit{(a)} to some node in \textit{(b)}.
    \begin{figure}
        \begin{subfigure}[t]{0.47\textwidth}
            \centering
            \begin{tikzpicture}
                \node[draw,circle] (a) {$u$};
                \node[draw,circle, below right=4cm and 1cm of a] (b) {$w$};
                \node[below left=0.6cm and 0.6cm of b] (cstar) {};
                \node[below left=0.8cm and 0.5cm of b] (cstar2) {};
                \node[below left=0.9cm and 0.3cm of b] (cstar3) {};
                \node[draw,circle, right=3cm of a,color=red] (c) {$v$};
                \draw (a) -- (b) node[pos=0.15,draw,circle,fill=white,scale=0.9,name=s1] {$s_{1}$} node[pos=0.5,draw,circle,fill=white] {\ldots} node[pos=0.85,draw,circle,fill=white,scale=0.6,name=sk] {$s_{k-3}$};
                \draw (b) -- (c) node[pos=0.15,draw,circle,fill=white,scale=0.6,name=s1b] {$s_{k-3}$} node[pos=0.5,draw,circle,color=red,fill=white] {\ldots} node[pos=0.85,draw,circle,color=red,fill=white,scale=0.9,name=s1c] {$s_1$};

                \node[draw,circle,below =1cm of b] (hangb) {};
                \node[above right=0.1cm and 1.5cm of hangb] (hangbr) {};
                \node[above left=0.1cm and 1.5cm of hangb] (hangbl) {};
                \node[below left=0cm and -0.1cm of hangbr,color=sidecolor] (hangblabel) {\tiny $\gamma$-\pitchfork};
                \draw[color=sidecolor] (hangbr) edge[dashed] (hangbl);

                \node[draw,circle,below=1cm of hangb] (hang) {};

                \node[draw,circle,below left=1cm and 1.1cm of hang,color=red] (a1) {};
                \node[draw,circle,below left=1cm and 0.3cm of hang,color=red] (a2) {};
                \node[below=0.9cm of hang] (am) {\ldots};
                \node[draw,circle,below right=1cm and 0.3cm of hang,color=red] (a3) {};
                \node[draw,circle,below right=1cm and 1.1cm of hang,color=red] (a4) {};

                \node[above=0cm of s1b] (s1bs) {};

                \draw [black,decorate,thick, decoration = {calligraphic brace, raise = 4pt, amplitude = 8pt,mirror}] (a.west) -- node[below left=-0.15cm and 0.4cm] {\small(a)}  (b.south west);
                \draw [black,decorate,thick, decoration = {calligraphic brace, raise = 4pt, amplitude = 8pt,mirror}] (hangb.west) -- node[above left=0.0cm and 0.4cm] {\small(b)}  (a1.west);

                \draw [black,decorate,thick, decoration = {calligraphic brace, raise = 4pt, amplitude = 8pt,mirror}] (s1b.south east) -- node[below right=-0.05cm and 0.4cm] {\small(c)}  (c.east);

                \draw (b) -- (hangb);
                \draw (hangb) -- (hang);
                \draw (a1) -- (hang);
                \draw (a2) -- (hang);
                \draw (a3) -- (hang);
                \draw (a4) -- (hang);
                \draw[opacity=0.4] (cstar) -- (b);
                \draw[opacity=0.4] (cstar2) -- (b);
                \draw[opacity=0.4] (cstar3) -- (b);
            \end{tikzpicture}
            \caption{Depiction of the different areas of the shortest path tree with the smallest number of hops. Notice that red nodes depict nodes at distance $k+1$.}
            \label{fig:shortest-path-tree}
        \end{subfigure}
        \hfill
        \begin{subfigure}[t]{0.47\textwidth}
            \centering
            \begin{tikzpicture}
                \node (middle) {};

                \node[draw,circle,scale=0.7,left=0.5cm of middle] (ab) {$v$};
                \node[draw,circle,scale=0.7,right=0.5cm of middle] (bc) {$u$};

                \node[draw,circle,below left=2cm and 0.8cm of middle] (a) {1};
                \node[draw,circle,below=1.92cm of middle] (b) {2};
                \node[below right=0.5cm and 0cm of b] (bb) {};
                \node[draw,circle,below right=2cm and 0.8cm of middle] (bcd) {3};
                \node[below right=0.5cm and 0cm of bcd] (bcdb) {};

                \draw (ab) -- (a);
                \draw (ab) -- (b);
                \draw (b) -- (bc);
                \draw (bcd) -- (bc);

                \draw[color=sidecolor] (bb) -- (b);
                \draw[color=sidecolor] (bcd) -- (bcdb);

                \node[draw,circle,below =1cm of a] (hangb) {};
                \node[draw,circle,below=1cm of hangb] (hang) {};
                \node[draw,circle,below=0.9cm of hang] (a1) {};
                \node[above=2.5cm of a1] (shadow-a1) {};
                \node[below right=1cm and 0.29cm of hang] (a3) {\ldots};
                \node[draw,circle,below right=1cm and 1.3cm of hang] (a4) {};
                \node[above right=0.1cm and 2cm of hangb] (hangbr) {};
                \node[above left=0.1cm and 0.5cm of hangb] (hangbl) {};
                \node[below left=0cm and -0.1cm of hangbr,color=sidecolor] (hangblabel) {\tiny $\gamma$-\pitchfork};

                \draw (hangbr) edge[dotted] (hangbl);

                \draw (a) -- (hangb);
                \draw (hangb) -- (hang);
                \draw (a1) -- (hang);
                \draw (a4) -- (hang);

                \draw (ab) edge[dashed] node[color=black,above] {\small(a)} (bc);
                \draw (b) edge[dashed] node[color=black,below] {\small(b)} (bcd);
                \draw (ab) edge[dashed, bend right] node[color=black,left] {\small(c)} (hangb);
                \draw (ab) edge[dashed, bend right] node[color=black,left] {\small(d)} (hang);
                \draw (ab) edge[dashed] node[color=black,above=0.15cm] {\small(e)} (bcd);

            \end{tikzpicture}
            \caption{Possible \complex shortcuts for $k=3$. Only a subgraph is shown $2$ and $3$ also have a $\gamma$-\pitchfork}
            \label{fig:shortest-path-tree-bet}
        \end{subfigure}
    \end{figure}

    \begin{enumerate}
        \item \textit{(a)$\leadsto$(a)} - Forms a $(k,\rho)$-ball for $u$ (as all nodes in the $\gamma$-\pitchfork can be reached from $u$ now in $k$ hops), while $v$ only reaches one node more than before in $k$ hops.
        \item \textit{(a)$\leadsto$(b)} - Forms a $(k,\rho)$-ball for $u$, but does not benefit $v$ at all, by \cref{lemma:tree} there has to be a shortest path tree where a part of \textit{(a)} is an ancestor of \textit{(b)}. For $v$ there is no such shortest path tree.
        \item \textit{(a)$\leadsto$(c)} - Only $1$ new node within the $\rho$-nearest neighbors is reachable for both $u$ and $v$, hence forbidden in an optimal solution by \cref{lemma:single-solution}.
        \item \textit{(b)$\leadsto$(c)} - Same as $\textit{(a)$\leadsto$(b)}$ but with $v$ forming a $(k,\rho)$-ball.
    \end{enumerate}
    For $k\geq 4$ there is no \complex shortcut forming more than one $(k,\rho)$-ball and forming none is ruled out by \cref{lemma:single-solution}.
    Now assume $k=3$ and that they share a common neighbor $w \in V$, otherwise as argued in \cref{lemma:single-short-cor} there is no shortcut which benefits more than one node and is part of an optimal solution.
    Now, \cref{fig:shortest-path-tree-bet} depicts the relevant case distinction for $k=3$.

    \begin{itemize}
        \item \textit{(a)} - Increases the number of nodes from the $\rho$-nearest neighbors reachable in $k$ hops of both $u$ and $v$ by at most $|V|+1$ for each one. This does not form a $(k,\rho'')$-ball for neither $v$ nor $u$ such that \cref{lemma:single-solution} forbids it.
        \item \textit{(b)} - This increases the number of nodes reached by at most $|V|+1$ for all edge nodes adjacent to $1$ or $2$ except for $u$, \cref{lemma:single-solution} again forbids this.
        \item \textit{(c), (d)} - Forms a $(k,\rho)$-ball but only for $v$, all other nodes adjacent to $2$ reach at most one new of their $\rho$-nearest neighbors in $k$ hops.
        \item \textit{(e)} - Forms a $(k,\rho)$-ball but only for $v$, assuming the $\gamma$-\pitchfork at $3$ is shortcutted. Observe that this shortcut does not benefit any other node apart from $v$.
    \end{itemize}
    Hence, for $k\geq 3$ any \complex shortcut in an optimal solution set results in exactly $1$ new $(k,\rho)$-ball.
    Any possible shortcut not discussed, is analogue to the ones shown here by switching $u$ and $v$.
    Due to the choice of $\gamma$ and by \cref{remark:shortcut-more-than-const} the constant increases by $1$ of $\rho$ nearest neighbors reached in $k$ hops can never add up enough such that it might become beneficial to use a different shortcut.
\end{proof}
\noindent

    \subsection{Decreasing \texorpdfstring{$\bm{\rho$}}{ρ}}\label{sec:appendix_blowup}
In \cref{section:kp-msp-np-hard}, we show a proof for \cref{theorem:kp-msp-hardness} that only considers $\rho = \gamma \ge 7|V| + 6$.
In the following, we discuss the necessary modifications to our construction to cover a larger regime of~$\rho$.
We reuse the established notation, namely, let $G = (V, E)$ be the input graph for \probVC, $\Gtrans = (\Vtrans, \Etrans)$ the graph obtained from the transformation in \cref{subsec:vc_transform}, and $\Gop = (\Vtrans, \Etrans \cup S)$ where $S$ is a minimal cardinality shortcut set to ensure $\Gop$ is a $(k, \rho)$ graph.

\begin{lemma}[blowup lemma] \label{lemma:const-tweak}
    We can reduce $\rho = \gamma$ such that $\gamma = \Oh(n^{\epsilon})$ with $\epsilon = \log_{N_c+y}\left(\sqrt{N_c}\right)\colon\;\forall y \in \mathbb{N}\colon\;y\leq N_c^\frac{D}{2}$, where $N_c=|\Vtrans|$ and $D$ is any constant integer.
\end{lemma}
\begin{proof}
    \begin{figure}
        \centering
        \rotatebox{90}{
            \begin{tikzpicture}
                \node (a) {};
                \node[draw,circle,below=1cm of a] (hangb) {};
                \node[draw,circle,below=1cm of hangb] (hang) {};
                \node[draw,circle,below left=1cm and 1.3cm of hang] (a1) {};
                \node[draw,circle,below left=1cm and 0.5cm of hang] (a2) {};
                \node[draw,circle,below=0.9cm of hang] (am) {};
                \node[draw,circle,below right=1cm and 0.5cm of hang] (a3) {};
                \node[draw,circle,below right=1cm and 1.3cm of hang] (a4) {};

                \node[below=0.9cm of am] (am2) {...};
                \node[draw,circle,below left=1cm and 1.3cm of am] (am2b) {};
                \node[draw,circle,below left=1cm and 0.5cm of am] (a2b) {};
                \node[draw,circle,below right=1cm and 0.5cm of am] (a3b) {};
                \node[draw,circle,below right=1cm and 1.3cm of am] (a4b) {};
                \draw (a) -- (hangb);
                \draw (hangb) -- (hang);
                \draw (a1) -- (hang);
                \draw (a2) -- (hang);
                \draw (am) -- (hang);
                \draw (a3) -- (hang);
                \draw (a4) -- (hang);

                \draw (am2b) -- (am);
                \draw (a2b) -- (am);
                \draw (a3b) -- (am);
                \draw (a4b) -- (am);
                \draw [black,decorate,thick, decoration = {calligraphic brace, raise = 4pt, amplitude = 8pt,mirror}] (am2b.south west) -- node[below right=0.5cm and 0.2cm,rotate=-90] {$y$}  (a4b.south east);
            \end{tikzpicture}}
        \caption{Modification to \pitchfork gadgets to decrease $\rho$.
            \label{fig:dep-dec}
        }
    \end{figure}
    As illustrated in \cref{fig:dep-dec}, attach $y$ new nodes to exactly one leaf in exactly one $\gamma$-\pitchfork.
    Since the number of nodes in the gadgets has to stay fixed we need to calculate $\rho$ in relation to the newly created graph.
    \begin{align*}
        \sqrt{N_c}             = (y+N_c)^\epsilon              \qquad \Leftrightarrow \qquad
        \log_{y+N_c}\sqrt{N_c}  = \epsilon
    \end{align*}
    The $\gamma$-\pitchfork has $\gamma+1$ nodes with a total diameter of $3$ such that any nodes attached at any of the nodes at the bottom will trivially form a $(k,\rho)$-ball, for every $k\geq 3$.
    Furthermore, any node of the newly attached nodes is at least $k+2$ hops away from any $(k,\rho)$-ball and by that, does not change the behavior of the shortcutting algorithm.
\end{proof}

\noindent
Finally, we will investigate the behavior of the function constricting $\epsilon$ that was introduced by \cref{lemma:const-tweak}.
\begin{lemma} \label{lemma:quantize-rho}
    The blowup lemma allows the construction of an arbitrary graph such that $\rho=\Oh(n^\epsilon)$ where $\frac{1}{D+1} < \epsilon < \frac{1}{D}$ provided that $D\geq 2$ is any constant integer.
\end{lemma}
\begin{proof}
    Recall $y \leq N_c^\frac{D}{2}$ where $D$ is any constant integer and $N_c$ is the number of nodes in $\Gtrans$. Setting $y = N_c^\frac{D}{2}$ allows us to derive bounds on $\epsilon$ for which the construction works in any case.
    \begin{align*}
        \epsilon                                                 & = \log_{y+N_c}\sqrt{N_c}                      \\
        \log_{N_c^\frac{D+1}{2}}\sqrt{N_c} < \epsilon            & < \log_{N_c^\frac{D}{2}}\sqrt{N_c}            \\
        \frac{\log\sqrt{N_c}}{\log N_c^\frac{D+1}{2}} < \epsilon & < \frac{\log\sqrt{N_c}}{\log N_c^\frac{D}{2}} \\
        \frac{1}{D+1} < \epsilon                                 & < \frac{1}{D}                                 \\
    \end{align*}
\end{proof}
\noindent
Observe that we proved the case of $\epsilon=\frac{1}{2}$ which is the largest $\epsilon$ for which the construction works and in \cref{lemma:const-tweak} we investigated all values for $\epsilon$ for which we are able to reduce to \textsc{Vertex cover} with our construction.
The final \cref{lemma:quantize-rho} formalizes the intuition of the limit behavior of $\epsilon$, that the rate at which we find new valid values of $\epsilon$ increases as we increase the number of nodes added in the blowup lemma.

\end{longversion}
\end{document}